\documentclass[11pt]{article}
\usepackage[margin=1in]{geometry} 
\usepackage{geometry}
\usepackage[linesnumbered,ruled]{algorithm2e} 	

\usepackage{float}					
\usepackage{scalerel}			
\usepackage{amsthm,amsfonts,amsmath,amssymb}
\usepackage{bm}						
\usepackage{color}
\usepackage{enumerate}
\usepackage{enumitem}				
\usepackage{tikz}			
\usetikzlibrary{arrows, shapes}
\usetikzlibrary{decorations.pathmorphing}
\usetikzlibrary{decorations.pathreplacing}
\usetikzlibrary {decorations.markings,fpu} 

\usepackage{caption}
\usepackage{mathtools}		
\usepackage{fancyhdr}		

\usepackage{bbm}
\definecolor{mygreen}{HTML}{009901}
\definecolor{toc}{RGB}{13,55,174}	
\usepackage{hyperref}				
\hypersetup{
    colorlinks=true,
    citecolor=red,
    filecolor=black,
    linkcolor=toc,
    urlcolor=toc
}

\allowdisplaybreaks 	

\usepackage{thm-restate}
\newtheorem{theorem}{Theorem}[section]
\newtheorem{lemma}[theorem]{Lemma}
\newtheorem{corollary}{Corollary}[section]

\newtheorem{mainLemma}[theorem]{Main Lemma}

\theoremstyle{definition}
\newtheorem{definition}[theorem]{Definition}

\newcommand{\e}{\varepsilon} 
\newcommand{\lp}{\left} 
\newcommand{\rp}{\right}
 
\renewcommand{\Pr}[2]{\textbf{Pr}_{#1}\lp[#2 \rp]}
\newcommand{\E}[2]{\mathbb{E}_{#1}{\lp[ #2\rp]}}
\newcommand{\reals}{\mathbb{R}}
\newcommand{\dist}{\mathcal{D}}
\newcommand{\ind}[1]{\mathbbm{1}{ \{ #1 \} }}

\newcommand{\opt}{\text{OPT}}

\newcommand{\scenarios}{\mathcal{S}}

\newcommand{\msscf}{\textsc{MSSC}_f}
\newcommand{\msscfc}{\textsc{MSSC}_f^c}
\newcommand{\umsscf}{\textsc{UMSSC}_f}
\newcommand{\umsscfc}{\textsc{UMSSC}_f^c}
\newcommand{\pbv}{\textsc{PB}}
\newcommand{\pb}{{\textsc{PB}}}
\newcommand{\pbT}{{\textsc{PB}_{\leq T}}}
\newcommand{\upbT}{{\textsc{UPB}_{\leq T}}}
\newcommand{\upbTc}{{\textsc{UPB}^c_{\leq T}}}
\newcommand{\pbvc}{{\textsc{PB}^c}}
\newcommand{\pbTc}{{\textsc{PB}_{\leq T}^c}}

\newcommand{\odt}{\textsc{DT}}
\newcommand{\odtc}{\textsc{DT}^c}

\newcommand{\upb}{{\textsc{UPB}}}

\newcommand{\udt}{\textsc{UDT}}
\newcommand{\udtc}{\textsc{UDT}^c}

\newcommand{\dsa}{{\mathcal{D}_{is_1}}}
\newcommand{\dsb}{{\mathcal{D}_{is_2}}}

\DeclareMathOperator{\ib}{IB}
\DeclareMathOperator{\DP}{DP}
\DeclareMathOperator{\ub}{UB}
\DeclareMathOperator{\nib}{NIB}
\DeclareMathOperator{\supp}{supp}
\DeclareMathOperator{\nat}{Nat}
\DeclareMathOperator{\sol}{Sol}
\DeclareMathOperator{\intv}{Int}

\newcommand{\iq}{{\mathcal{I}_q}}

\newcommand{\pbText}{\textsc{Pandora's Box}}
\newcommand{\dtText}{\textsc{Optimal Decision Tree}}
\newcommand{\udtText}{\textsc{Uniform Decision Tree}}
\newcommand{\msscText}{\textsc{Min Sum Set Cover with Feedback}}

\newcommand{\shuchiremove}[1]{ }

\newcommand{\cost}{\mathrm{cost}}

\title{Approximating Pandora's Box with Correlations\footnote{This work was funded in part by NSF awards CCF-2225259 and CCF-2217069}}
\date{}
\author{
Shuchi Chawla \\ UT-Austin \\ {\tt shuchi@cs.utexas.edu } \and 
Evangelia Gergatsouli \\ UW-Madison \\ {\tt evagerg@cs.wisc.edu} \and 
Jeremy McMahan \\ UW-Madison \\ {\tt jmcmahan@wisc.edu} \and 
Christos Tzamos \\ UW-Madison \&  University of Athens\\ {\tt tzamos@wisc.edu} 
}

\begin{document}

\maketitle

We revisit the classic Pandora’s Box ($\pbv$) problem under correlated distributions on the box values. Recent work of~\cite{ChawGergTengTzamZhan2020} obtained constant approximate algorithms for a restricted class of policies for the problem that visit boxes in a fixed order. In this work, we study the complexity of approximating the optimal policy which may adaptively choose which box to visit next based on the values seen so far.

Our main result establishes an approximation-preserving equivalence of $\pbv$ to the well studied Uniform Decision Tree ($\udt$) problem from stochastic optimization and a variant of the \textsc{Min-Sum Set Cover} ($\msscf$) problem. For distributions of support $m$, $\udt$ admits a $\log m$ approximation, and while a constant factor approximation in polynomial time is a long-standing open problem,  constant factor approximations are achievable in subexponential time~\cite{LiLianMuss2020}. Our main result implies that the same properties hold for $\pbv$ and $\msscf$.

We also study the case where the distribution over values is given more succinctly as a mixture of $m$ product distributions. This problem is again related to a noisy variant of the Optimal Decision Tree which is significantly more challenging. We give a constant-factor approximation
that runs in time $n^{ \tilde O( m^2/\e^2 ) }$ when the mixture components on every box are either identical or separated in TV distance by $\e$.

\setcounter{page}{0}
\thispagestyle{empty}
\newpage
\section{Introduction}\label{sec:intro}

Many everyday tasks involve making decisions under uncertainty; for example
driving to work using the fastest route or buying a house at the best price.
Although we don't know how the future outcomes of our current decisions will turn out, we can often use some prior information to
facilitate the decision making process. For example, having driven on the
possible routes to work before, we know which is usually the busiest
one. It is also common in such cases that we can remove part of the
uncertainty by paying some additional cost. This type of online decision making in the presence of costly information can be modeled as the so-called
\pbText{} problem, first formalized by Weitzman in \cite{Weit1979}. In this
problem, the algorithm is given $n$ alternatives called \emph{boxes}, each
containing a value from a known distribution. The exact value is not known, but
can be revealed at a known \emph{opening} cost specific to the box. The goal
of the algorithm is to decide which box to open next and whether to
select a value and stop, such that the total \emph{opening cost plus the
minimum value revealed} is minimized.  In the case of independent distributions
on the boxes' values, this problem has a very elegant and simple optimal solution, as
described by Weitzman \cite{Weit1979}: calculate
an index for each box\footnote{This is a special case of Gittins index
\cite{Gitt1972}.}, open the boxes in increasing order of index, and stop when the
expected gain is worse than the value already obtained.

Weitzman's model makes the crucial assumption that the distributions
on the values are independent across boxes. This, however, is not always the case
in practice and as it turns out, the simple algorithm of the
independent case fails to find the optimal solution under correlated
distributions. Generally, the complexity of the \pbText{} with
correlations is not yet well understood. \textbf{In this work we develop the first approximately-optimal policies for the \pbText{} problem with correlated values.}


We consider two standard models of correlation where the distribution over values can be specified explicitly in a succinct manner. In the first, the distribution over values has a small support of size $m$. In the second the distribution is a mixture of $m$ product distributions, each of which can be specified succinctly. We present approximations for both settings.

A primary challenge in approximating \pbText{} with correlations is that the optimal solution can be an adaptive policy that determines which box to open depending on the instantiations of values in all of the boxes opened previously. It is not clear that such a policy can even be described succinctly. Furthermore, the choice of which box to open is complicated by the need to balance two desiderata -- finding a low value box quickly versus learning information about the values in unopened boxes (a.k.a. the state of the world or realized scenario) quickly. Indeed, the value contained in a box can provide the algorithm with crucial information about other boxes, and inform the choice of which box to open next; an aspect that is completely missing in the independent values setting studied by Weitzman.

\subsubsection*{Contribution 1: Connection to Decision Tree and a general purpose approximation.}

Some aspects of the \pbText{} problem have been studied separately in other contexts. For example, in the \dtText{} problem ($\odt$)~\cite{GuilBilm2009, LiLianMuss2020}, the goal is to identify an unknown hypothesis, out of $m$ possible ones, by performing a sequence of costly tests, whose outcomes depend on the realized hypothesis. This problem has an informational structure similar to that in \pbText{}. In particular, we can think of every possible joint instantiation of values in boxes as a possible hypothesis, and every opening of a box as a test. The difference between the two problems is that while in \dtText{} we want to identify the realized hypothesis exactly, in \pbText{} it suffices to terminate the process as soon as we have found a low value box. 

Another closely related problem is the \textsc{Min Sum Set Cover}~\cite{FeigLovaTeta2004}, where boxes only have two kinds of values -- {\em acceptable} or {\em unacceptable} -- and the goal is to find an {\em acceptable} value as quickly as possible. A primary difference relative to \pbText{} is that unacceptable boxes provide no further information about the values in unopened boxes. 

One of the main contributions of our work is to unearth connections between \pbText{} and the two problems described above. We show that \pbText{} is essentially equivalent to a special case of \dtText{} (called \textsc{Uniform Decision Tree} or UDT) where the underlying distribution over hypotheses is uniform -- the approximation ratios of these two problems are related within log-log factors. 
Surprisingly, in contrast, the non-uniform $\odt$ appears to be harder than non-uniform \pbText{}.
We relate these two problems by showing that both are in turn related to a new version of \textsc{Min Sum Set Cover}, that we call \msscText{} ($\msscf$). These connections are summarized in Figure~\ref{fig:intro_summary}. 
We can thus build on the rich history and large collection of results on these problems to offer efficient algorithms for \pbText{}.
We obtain a polynomial time $\tilde{O}(\log m)$ approximation for \pbText{}, where $m$ is the number of distinct value vectors (a.k.a. scenarios) that may arise; as well as constant factor approximations in subexponential time. 
\begin{figure}[H]
    \centering
    \pgfmathsetmacro{\dist}{4}
\pgfmathsetmacro{\distV}{2}
\pgfmathsetmacro{\distVerBrace}{0.5}
\begin{tikzpicture}

	\node (pbv) at (0,0){{\LARGE $\pbv$}};
 
 
	\node (umsscf) at (\dist,-0.05){{\LARGE $\umsscf$}};

	\node[opacity=0] (dummyumsscf) at (\dist,0){{\LARGE $\umsscf$}};

	\node (udt) at (2*\dist,0){{\LARGE $\udt$}};

	 \draw[<->,thick] (pbv) to node [above] {Section \ref{sec:pb_to_mssc}} (dummyumsscf);

		 \draw[<->,thick] (udt) to node[above] {Section~\ref{sec:mssc_to_dt}}(dummyumsscf);

	 \draw[decorate,black,thick, decoration={brace,amplitude=6pt, mirror}] (0,-\distVerBrace) -- node[below, yshift=-0.2cm]{Log-log factors} (\dist-0.2,-\distVerBrace);
  
    \draw[decorate,black,thick, decoration={brace,amplitude=6pt, mirror}] (\dist+0.2,-\distVerBrace) -- node[below, yshift=-0.2cm]{Constant factors} (2*\dist,-\distVerBrace);
\end{tikzpicture}
    \caption{A summary of our approximation preserving reductions 
    }
    \label{fig:intro_summary}
\end{figure}
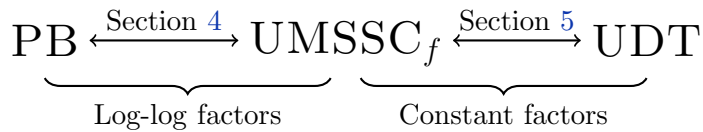
It is an important open question whether constant factor approximations exist for \udtText{}: the best known lower-bound on the approximation ratio is $4$ while it is known that it is not NP-hard to obtain super-constant approximations under the Exponential Time Hypothesis. The same properties transfer also to \pbText{} and \msscText{}. Pinning down the tight approximation ratio for any of these problems will directly answer these questions for any other problem in the equivalence class we establish.

The key technical component in our reductions is to find an appropriate stopping rule for \pbText{}: after opening a few boxes, how should the algorithm determine whether a small enough value has been found or whether further exploration is necessary? We develop an iterative algorithm that in each phase finds an appropriate threshold, with the exploration terminating as soon as a value smaller than the threshold is found, such that there is a constant probability of stopping in each phase. Within each phase then the exploration problem can be solved via a reduction to $\udt$. The challenge is in defining the stopping thresholds in a manner that allows us to relate the algorithm's total cost to that of the optimal policy.

\subsubsection*{Contribution 2: Approximation for the mixture of distributions model.}

Having established the general purpose reductions between \pbText{} and $\odt$, we turn to the mixture of product distributions model of correlation. This special case of \pbText{} interpolates between Weitzman's independent values setting and the fully general correlated values setting. In this setting, we use the term ``scenario" to denote the different product distributions in the mixture. The information gathering component of the problem is now about determining which product distribution in the mixture the box values are realized from. Once the algorithm has determined the realized scenario (a.k.a. product distribution), the remaining problem amounts to implementing Weitzman's strategy for that scenario.

We observe that this model of correlation for \pbText{} is 
related to the noisy version of $\odt$, where the results of some tests for a given realized hypothesis are not deterministic. One challenge for $\odt$ in this setting is that any individual test may give us very little information distinguishing different scenarios, and one needs to combine information across sequences of many tests in order to isolate scenarios. This challenge is inherited by \pbText{}. 

Previous work on noisy $\odt$ obtained algorithms whose
approximations and runtimes depend on the amount of noise. In contrast, we consider settings where the level of noise is arbitrary, but where the mixtures satisfy a {\em separability assumption}. In particular, we assume that for any given box, if we consider the marginal distributions of the value in the box under different scenarios, these distributions are either identical or sufficiently different (e.g., at least $\e$ in TV distance) across different scenarios. Under this assumption, we design a constant-factor approximation for \pbText{} that runs
in $n^{\tilde{O}(m^2/\e^2)}$ (Theorem~\ref{thm:dp}), where $n$ is the number of
boxes.
The formal result and the algorithm is presented in Section~\ref{sec:mixt}.




\shuchiremove{
That is,  the value distributions are arbitrarily correlated and the policies under consideration are not restricted in any way but are allowed to exploit the feedback obtained in previous steps arbitrarily. We unravel a deep connection of this problem with two other well studied problems in approximation algorithms and stochastic optimization, namely \dtText{} ($\odt$) and a variant of \msscText{} ($\msscf$),  which allows us to obtain the first non-trivial approximation algorithms for the general \pbText{} problem.
\begin{itemize}
      \item 
  In \dtText, we are asked to identify an
unknown hypothesis, out of $m$ possible ones, by performing a sequence of
tests. Each test has a cost and, if chosen, reveals a result, which depends on
the hypothesis realized. The goal of the algorithm is to minimize the
expected cost of tests performed in order to identify the correct
hypothesis given an initial prior over the set of hypotheses. The special case where this prior is uniform corresponds to \udtText. This problem has found many
applications in medical diagnosis~\cite{PodgKokoStigRozm}, fault
diagnosis~\cite{PattDont1992}, and active learning~\cite{GuilBilm2009,GoloKrau2010}. 
\item 
In \textsc{Min Sum Set Cover}~\cite{FeigLovaTeta2004}, there is a known ground set of elements and an unknown set drawn from some known probability distribution over a collection of sets. The goal of the algorithm is to sequentially select elements at a cost so that the expected cost paid until the set is hit, i.e. contains the chosen element, is minimized. In the \msscText{} that we consider, every element chosen provides feedback which depends on the realized set. This enables the algorithm to choose the next element adaptively.
  \end{itemize}

  Our main result shows how to readily apply algorithms for the two problems above to obtain approximation algorithms for \pbText{} with a small loss in the approximation ratio. We can thus build on the rich history and large collection of results on these problems to offer efficient algorithms for \pbText{}. 
}

\shuchiremove{
A more detailed overview of our results is given in the next section.

\subsection{Our Results}
Our work establishes an equivalence between \pbText{}, \udtText{} and \msscText{}. We achieve this through a series of approximation preserving reductions  summarized in Figure~\ref{fig:intro_summary}.

\begin{figure}[H]
    \centering
    \pgfmathsetmacro{\dist}{4}
\pgfmathsetmacro{\distV}{2}
\pgfmathsetmacro{\distVerBrace}{0.5}
\begin{tikzpicture}

	\node (pbv) at (0,0){{\LARGE $\pbv$}};
 
 
	\node (umsscf) at (\dist,-0.05){{\LARGE $\umsscf$}};

	\node[opacity=0] (dummyumsscf) at (\dist,0){{\LARGE $\umsscf$}};

	\node (udt) at (2*\dist,0){{\LARGE $\udt$}};

	 \draw[<->,thick] (pbv) to node [above] {Section \ref{sec:pb_to_mssc}} (dummyumsscf);

		 \draw[<->,thick] (udt) to node[above] {Section~\ref{sec:mssc_to_dt}}(dummyumsscf);

	 \draw[decorate,black,thick, decoration={brace,amplitude=6pt, mirror}] (0,-\distVerBrace) -- node[below, yshift=-0.2cm]{Log-log factors} (\dist-0.2,-\distVerBrace);
  
    \draw[decorate,black,thick, decoration={brace,amplitude=6pt, mirror}] (\dist+0.2,-\distVerBrace) -- node[below, yshift=-0.2cm]{Constant factors} (2*\dist,-\distVerBrace);
\end{tikzpicture}
    \caption{High level summary of results \textcolor{red}{this is a new figure lmk if you like it}}
    \label{fig:intro_summary}
\end{figure}



This equivalence leads to a sequence of corollaries on \pbText{} by exploiting known results for \dtText{} \cite{GuilBilm2009, LiLianMuss2020} and \msscText{} \cite{FeigLovaTeta2004}. In particular, we show that one can compute in polynomial time an approximately-optimal policy for \pbText{} whose approximation ratio is logarithmic in the support of the distribution of values. Moreover, we also show how to obtain constant factor approximations in subexponential time. 

It is an important open question whether constant factor approximations exist for \udtText{}: the best known lower-bound on the approximation ratio is 4 while it is known that it is not NP-hard to obtain super-constant approximations under the Exponential Time Hypothesis. The same properties transfer also to \pbText{} and \msscText{}. Pinning down the tight approximation ratio for any of this problems will directly answer these questions for any other problem in the equivalence class we established.

Finally, through our equivalence we are able to recover the results of ~\cite{ChawGergTengTzamZhan2020} which studied the versions of correlated Pandora's box with limited adaptivity through the direct connection with the standard \textsc{Min Sum Set Cover} (without feedback) that we establish. This allows us to address even the case where boxes have arbitrary costs, left open by the work of ~\cite{ChawGergTengTzamZhan2020}. The results and implications are outlined in more detail in Section~\ref{sec:roadmap}.

We also study the case where we are given a mixture of $m$ product distributions, instead of explicitly given distributions. In this case \pbText{} is
related to the noisy version of the optimal decision tree. The noise can be
arbitrary and we only require the mixtures to satisfy a \emph{separability}
condition; the marginals should either be identical or differ by at least $\e$ in TV distance. Using
this property, we design a constant-factor approximation for $\pbT$ that runs
in $n^{\tilde{O}(m^2/\e^2)}$ (Theorem~\ref{thm:dp}), where $n$ is the number of
boxes, which also implies a constant-factor for the initial \pbText{}
problem with values (Corollary~\ref{cor:mixt}), when using the tool of
Theorem~\ref{thm:pbv_to_pb0_log}. The formal result and the algorithm is presented in Section~\ref{sec:mixt}.
}
 \subsection{Related work}
	The \pbText{} problem was first introduced by Weitzman in the Economics literature
\cite{Weit1979}. Since then, there has been a long line of research studying
\pbText{} and its many variants ; non-obligatory inspection~\cite{Dova2018,BeyhKlein2019,BeyhCai2022,FuLiLiu2022}, with order constraints~\cite{HazoAumaKrauSarn2013,BoodFuscLazoLeon2020},
with correlation~\cite{ChawGergTengTzamZhan2020,GergTzam2023},
with combinatorial costs~\cite{BergEzraFeld2023},
competitive information design~\cite{DingFengHoXu2023},
delegated version~\cite{BechDughPate2022},
and finally in an online setting~\cite{EsfaHajiLuciMitz2019}. 
Multiple works also study the generalized setting where more information can be obtained for a
price~\cite{CharFagiGuruKleiRaghSaha2002, GuptKuma2001,
ChenJavdKarbBagnSrinKrau2015, ChenHassKarbKrau2015} and in settings
with more complex combinatorial constraints~\cite{Sing2018,
GoelGuhaMuna2006, GuptNaga2013, AdamSvirWard2016, GuptNagaSing2016,
GuptNagaSing2017, GuptJianSing2019}.

Chawla et al. \cite{ChawGergTengTzamZhan2020} were the first to study \pbText{} with correlated values, but they designed approximations relative to a simpler benchmark, namely the optimal performance achievable using a so-called {\em Partially Adaptive} strategy that cannot adapt the order in which it opens boxes to the values revealed. In general, optimal strategies can decide both the ordering of the boxes and the stopping time based on the values revealed. \cite{ChawGergTengTzamZhan2020} designed an algorithm with performance no more than a constant factor worse than the optimal Partially Adaptive strategy. 

In \textsc{Min Sum Set Cover} the line of work was initiated by \cite{FeigLovaTeta2004}, and continued with improvements and generalizations to more complex constraints by~\cite{AzarGamzYin2009,
MunaBabuMotwWido2005, BansGuptRavi2010, SkutWill2011}.

Optimal decision tree is an old problem studied in a variety of settings (\cite{PodgKokoStigRozm,PattDont1992,GuilBilm2009,GoloKrau2010}), while
its most notable application is in active learning settings. It was proven to
be NP-Hard by Hyafil and Rivest \cite{HyafRive1976}. Since then the problem of
finding the best approximation algorithm was an active one
\cite{GareGrah1974,Love1985,KosaPrzyBorg1999,Dasg2004, ChakPandRoyAwasMoha2011,
ChakPandRoySabh2009,GuilBilm2009,GuptNagaRavi2017,CicaJacoLabeMoli2010,AdleHeer2011},
where finally a greedy $\log m$ for the general case was given by
\cite{GuilBilm2009}. This approximation ratio is proven to be the best possible
\cite{ChakPandRoyAwasMoha2011}. For the case of Uniform decision tree less is
known, until recently the best algorithm was the same as the optimal decision
tree, and the lower bound was $4$ \cite{ChakPandRoyAwasMoha2011}. The recent
work of Li et al. \cite{LiLianMuss2020} showed that there is an algorithm
strictly better than $\log m$ for the uniform decision tree.

The noisy version of optimal decision tree was first studied in
\cite{GoloKrau2010}\footnote{This result is based on a result from
\cite{GoloKrau2011} which turned out to be wrong \cite{NanSali2017}. The
correct results are presented in \cite{GoloKrau2017correct}}, which gave an
algorithm with runtime that depends exponentially on the number of noisy
outcomes. Subsequently, Jia et al. in
\cite{JiaNagaNaviRavi2019} gave an $(\min(r,h) + \log m)$-approximation
algorithm, where $r$ (resp. $h$) is the maximum number of different test
results per test (resp. scenario) using a reduction to Adaptive Submodular
Ranking problem \cite{KambNagaNavi2017}. In the case of large number of
noisy outcome they obtain a $\log m$ approximation exploiting the connection to
Stochastic Set Cover \cite{LiuPartRangYang2008,ImNagaZwaa2016}.

\section{Preliminaries}\label{sec:prelims}
In this paper we study the connections between three different sequential decision making problems -- \dtText, \pbText{}, and \textsc{Min Sum Set Cover}. We describe these problems formally below.

\subsubsection*{Optimal Decision Tree}
In the \dtText{} problem (denoted $\odt$) we are given a set $\scenarios$ of $m$ scenarios
$s\in \scenarios$, each occurring with (known) probability $p_s$; and $n$ tests $\mathcal{T}=\{T_i\}_{i\in [n]}$, each with cost $1$. Nature picks a
scenario $s\in \scenarios$ from the distribution $p$ but this scenario is unknown to the algorithm. The goal of the algorithm is to determine which scenario is realized by running a subset of the tests $\mathcal{T}$. When test $T_i$ is run and the realized scenario is $s$, the test returns a result $T_i(s) \in \reals$.

\paragraph{Output.} The output of the algorithm is a decision tree where at each node there is a
test that is performed, and the branches are the outcomes of the test. In each
of the leaves there is an individual scenario that is the only one consistent
with the results of the test in the unique path from the root to this leaf. Observe that there is a single leaf corresponding to each scenario $s$. We can represent the tree as an \emph{adaptive policy} defined as follows:

\begin{definition}[Adaptive Policy $\pi$]
An adaptive policy $\pi: \cup_{X\subseteq \mathcal{T}} \reals^{X} \rightarrow \mathcal{T}$ is a function that given a set of tests done so far and their results, returns the next test to be	performed. 
\end{definition}

\paragraph{Objective.} For a given decision tree or policy $\pi$, let $\cost_s(\pi)$ denote the total cost of all of the tests on the unique path in the tree from the root to the leaf labeled with scenario $s$. The objective of the algorithm is to find a policy $\pi$ that minimizes the average cost $\sum_{s\in
\scenarios} p_s \cost_s(\pi)$.

\noindent
We use the term \udtText{} ($\udt$) to denote the special case of the problem where $p_s = 1/m$  for all scenarios $s$.

\subsubsection*{Pandora's Box}
In the \pbText{} problem we are given $n$ boxes, each with cost $c_i\geq 0$ and value $v_i$. The values $\{v_i\}_{i\in [n]}$ are distributed according to known distribution $\mathcal{D}$. We assume that $\mathcal{D}$ is an arbitrary correlated distribution over vectors $\{v_i\}_{i\in [n]}\in\reals^n$. We call vectors of values {\em scenarios} and use $s=\{v_i\}_{i\in [n]}$ to denote a possible realization of the scenario. As in $\odt$, nature picks a scenario from the distribution $D$ and this realization is a priori unknown to the algorithm. The goal of the algorithm is to pick a box of small value. The algorithm can observe the values realized in the boxes by opening any box $i$ at its respective costs $c_i$.

\paragraph{Output.}The output of the algorithm is an adaptive policy $\pi$ for opening boxes and a stopping condition. The policy $\pi$ takes as input a subset of the boxes and their associated values, and either returns the index of a box $i\in [n]$ to be opened next or stops and selects the minimum value seen so far. That is, $\pi:\cup_{X\subseteq [n]} \reals^X\rightarrow [n]\cup\{\perp\}$ where $\perp$ denotes stopping.

\paragraph{Objective.} For a given policy $\pi$, let $\pi(s)$ denote the set of boxes opened by the policy prior to stopping when the realized scenario is $s$. The objective of the algorithm is to minimize the expected cost of the boxes opened plus the minimum value discovered, where the expectation is
taken over all possible realizations of the values in each box.\footnote{In the original version of the problem studied by Weitzman~\cite{Weit1979} the values are independent across boxes, and the goal is to maximize the value collected minus the costs paid, in contrast to the minimization version we study here.}  Formally the objective is given by
\[
	\E{s\sim \mathcal{D}}{\min_{i\in \pi(s)} v_{is} + \sum_{i\in\pi(s)} c_i},
\]

For simplicity of presentation, from now on we assume that $c_i=1$ for all boxes, but we show in Section~\ref{apn:general_costs} how to adapt our results to handle non-unit costs, without any loss in the approximation factors.

We use $\upb$ to denote the special case of the problem where the distribution $\mathcal{D}$ is uniform over $m$ scenarios.

\subsubsection*{Min Sum Set Cover with Feedback}
In Min Sum Set Cover, we are given $n$ elements and a collection of $m$ sets $\scenarios$ over them, and a distribution $\mathcal{D}$ over the sets. The output of the algorithm is an ordering $\pi$ over the elements. The cost of the ordering for a particular set $s\in\scenarios$ is the index of the first element in the ordering that belongs to the set $s$, that is, $\cost_s(\pi)=\min\{i:\pi(i)\in s\}$. The goal of the algorithm is to minimize the expected cost $\E{s\sim\mathcal{D}}{\cost_s(\pi)}$.

We define a variant of the Min Sum Set Cover problem, called \msscText{} ($\msscf$). As in the original problem, we are given a set of $n$ elements, a collection of $m$ sets $\scenarios$ and a distribution $\dist$ over the sets. Nature instantiates a set $s\in\scenarios$ from the distribution $\dist$; the realization is unknown to the algorithm. Furthermore, in this variant, each element provides {\em feedback} to the algorithm when the algorithm "visits" this element; this feedback takes on the value $f_i(s) \in \reals$ for element $i\in [n]$ if the realized set is $s\in \scenarios$. 

\paragraph{Output.} The algorithm once again produces an ordering $\pi$ over the elements. Observe that the feedback allows the algorithm to adapt its ordering to previously observed values. Accordingly, $\pi$ is an adaptive policy that maps a subset of the elements and their associated feedback, to the index of another element $i\in [n]$. That is, $\pi:\cup_{X\subseteq [n]} \reals^X\rightarrow [n]$.

\paragraph{Objective.} As before, the cost of the ordering for a particular set $s\in\scenarios$ is the index of the first element in the ordering that belongs to the set $s$, that is, $\cost_s(\pi)=\min\{i:\pi(i)\in s\}$. The goal of the algorithm is to minimize the expected cost $\E{s\sim\mathcal{D}}{\cost_s(\pi)}$.

\subsubsection*{Commonalities and notation} As the reader has observed, we capture the commonalities between the different problems through the use of similar notation. Scenarios in $\odt$ correspond to value vectors in $\pb$ and to sets in $\msscf$; all are denoted by $s$, lie in the set $\scenarios$, and are drawn by nature from a known joint distribution $\dist$. Tests in $\odt$ correspond to boxes in $\pb$ and elements in $\msscf$; we index each by $i\in [n]$. The algorithm for each problem produces an adaptive ordering $\pi$ over these tests/boxes/elements. Test outcomes $T_i(s)$ in $\odt$ correspond to box values $v_{i}(s)$ in $\pb$ and feedback $f_i(s)$ in $\msscf$. We will use the terminology and notation across different problems interchangeably in the rest of the paper.

\subsection{Modeling Correlation}
	In this work we study two general ways of modeling the correlation between the
	values in the boxes. 
	\paragraph{Explicit Distributions.} In this case, $\mathcal{D}$ is a
	distribution over $m$ \emph{scenarios} where the $j$'th scenario is
	realized with probability $p_j$, for $j\in[m] $.  Every scenario
	corresponds to a fixed and known vector of values contained in each box.
	Specifically, box $i$ has value $v_{ij}\in \mathbb{R}^+ \cup \{ \infty \}$
	for scenario $j$. 

	\paragraph{Mixture of Distributions.} We also consider a more general
	setting, where $\mathcal{D}$ is a mixture of $m$
	product distributions. Specifically, each scenario $j$ is a product
	distribution; instead of giving a deterministic value for every box $i$,
	the result is drawn from distribution $\mathcal{D}_{ij}$. This setting is a
	generalization of the explicit distributions setting described before.

\section{Roadmap of the Reductions and Implications}\label{sec:roadmap}
In Figure~\ref{fig:summary}, we give an overview of all the main technical reductions shown in Sections~\ref{sec:pb_to_mssc} and \ref{sec:mssc_to_dt}.
An arrow $A\rightarrow B$ means that we gave an approximation preserving reduction from problem $A$ to problem $B$. 
Therefore an algorithm for $B$ that achieves approximation ratio $\alpha$ gives also an algorithm for $A$ with approximation ratio $O(\alpha)$ (or $O(\alpha \log \alpha)$ in the case of black dashed lines). For the exact guarantees we refer to the formal statement of the respective theorem. The gray lines denote less important claims or trivial reductions (e.g. in the case of $A$ being a subproblem of $B$).

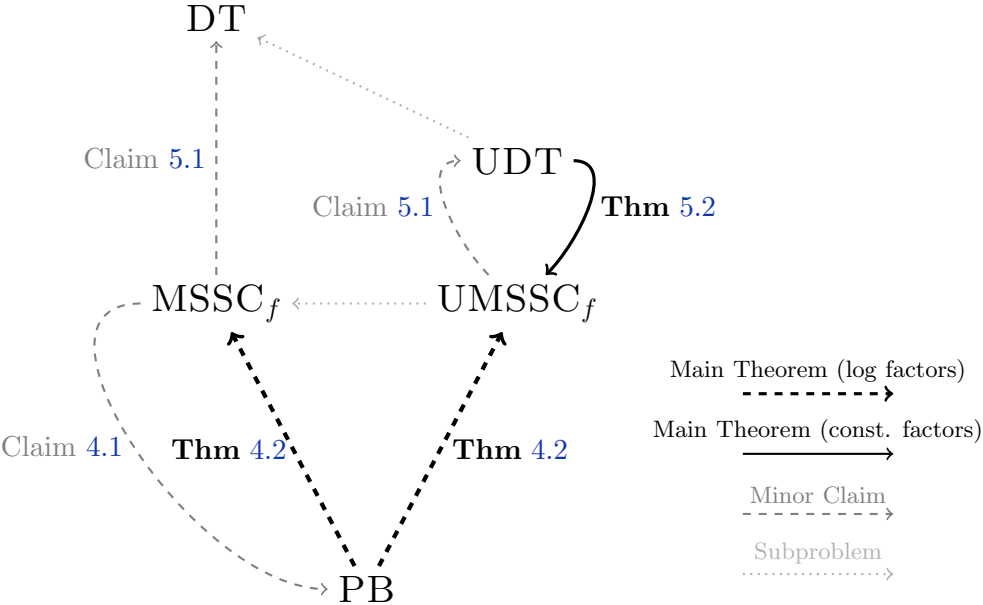
\begin{figure}[H]
    \centering
    \newcommand{\bl}[1]{\textcolor{black}{#1}}
\pgfmathsetmacro{\dist}{3}
\pgfmathsetmacro{\distV}{1.9}
\pgfmathsetmacro{\offset}{1}
\pgfmathsetmacro{\ypomnimaX}{6}
\pgfmathsetmacro{\ypomnimaY}{1.8}
\begin{tikzpicture}

	\node (pbv) at (0,0){{\Large $\pbv$}};

	\node (umsscf) at (2*\offset,2*\distV){{\Large $\umsscf$}};
 	\node (msscf) at (-2*\offset,2*\distV){{\Large $\msscf$}};

	\node (udt) at (2*\offset, 3*\distV){{\Large $\udt$}};
 	\node (odt) at (-2*\offset, 4*\distV){{\Large $\odt$}};
	\draw[->,dotted, gray, thick, opacity=0.6] (umsscf) -- (msscf);

\draw[->,dotted, gray, thick, opacity=0.6] (udt) --  (odt);

    \draw[->, dashed, thick, gray] (msscf) to [out=180, in=180]  node [above, left] {Claim~\ref{cl:msscf_to_pb}}  (pbv);
 
    \draw[->,dashed, thick, gray] (msscf) -- node [above, left] {Claim~\ref{thm:msscf_to_odt}} (odt); 
    
    \draw[->, dashed, thick, gray] (umsscf) to [out=135,in=180]  node [above, left] {Claim~\ref{thm:msscf_to_odt}}  (udt);

    \draw[->,ultra thick, dashed] (pbv) -- node [above, right, text width=2cm] {\textbf{Thm}~\ref{thm:pb_to_mssc}}(umsscf);
        
    \draw[->,very thick] (udt) to [in=45, out=0] node [above, right] {\textbf{Thm}~\ref{thm:udt_to_pbT}}  (umsscf); 
    
     \draw[->,ultra thick, dashed] (pbv) -- node [label={[xshift=-0.85cm, yshift=-0.45cm]\textbf{Thm}~\ref{thm:pb_to_mssc}}] {} (msscf);

    \draw[->,very thick, dashed] (\ypomnimaX-1, \ypomnimaY+0.8) -- node[above] {{\footnotesize Main Theorem ($\log$ factors)}} (\ypomnimaX+1, \ypomnimaY+0.8);
    
        \draw[->,thick] (\ypomnimaX-1, \ypomnimaY) -- node[above] {{\footnotesize Main Theorem (const. factors)}} (\ypomnimaX+1, \ypomnimaY);
        
    \draw[->,dashed,thick, gray] (\ypomnimaX-1, \ypomnimaY-0.8) -- node[above] {{\footnotesize Minor Claim}} (\ypomnimaX+1, \ypomnimaY-0.8);
    
	  \draw[->,dotted, gray, thick, opacity=0.6] (\ypomnimaX-1, \ypomnimaY-1.6) -- node[above] {{\footnotesize Subproblem}} (\ypomnimaX+1, \ypomnimaY-1.6);

\end{tikzpicture}
    \caption{Summary of all our reductions. Bold black lines denote our main theorems, gray dashed are minor claims, and dotted lines are trivial reductions.}
    \label{fig:summary}
\end{figure}

\subsection{Approximating  \pbText{}}
Given our reductions and 
using the best known results for \udtText{} from~\cite{LiLianMuss2020} we immediately obtain efficient approximation algorithms for \pbText. We repeat the results of~\cite{LiLianMuss2020} below.

\begin{theorem}[Theorems 3.1 and 3.2 from \cite{LiLianMuss2020}]\label{thm:old_thm}
\,
\begin{itemize}
\item 
There is a $O(\log m/\log \mathrm{OPT})$-approximation algorithm for $\udt{}$ that runs in polynomial time, where $\opt$ is the cost of the optimal solution of the $\udt$ instance. 
\item There is a $\frac{9+\e}{\alpha}$-approximation algorithm for $\udt{}$ that runs in time $n^{\tilde{O}(m^\alpha)}$ for any $\alpha \in (0,1)$. 
\end{itemize}
\end{theorem}

Using the results of
Theorem~\ref{thm:old_thm} combined with Theorem~\ref{thm:pb_to_mssc} and Claim~\ref{thm:msscf_to_odt} we get the following corollary.

\begin{corollary}\label{cor:odt_tilde_logm}
  From the best-known results for $\udt$, we have that
\begin{itemize}
\item 
There is a $\tilde O(\log m)$-approximation algorithm for $\pbv$ that runs in polynomial time\footnote{If additionally the possible number of outcomes is a constant $K$, this gives a $O(\log m)$ approximation without losing an extra logarithmic factor, since $\opt \geq \log_K m$, as observed by~\cite{LiLianMuss2020}.}.
\item There is a $\tilde{O}(1/\alpha)$-approximation algorithm for $\pbv$ that runs in time $n^{\tilde{O}(m^\alpha)}$ for any $\alpha \in (0,1)$. 
\end{itemize}
\end{corollary}

An immediate implication of the above corollary is that it is not NP-hard to obtain a superconstant  approximation for $\pbv$, formally stated below.  
\begin{corollary}
    It is not NP-hard to achieve any superconstant approximation for $\pbv$ assuming the Exponential Time Hypothesis.
\end{corollary}

Observe that the logarithmic approximation achieved in Corollary~\ref{cor:odt_tilde_logm} loses a $\log \log m$ factor (hence the $\tilde{O}$) as it relies on the more complex reduction of Theorem~\ref{thm:pb_to_mssc}. If we choose to use the more direct reduction of Theorem~\ref{thm:pbv_to_pb0} to the \dtText{} where the tests have non-unit costs (which also admits a $O(\log m)$-approximation~\cite{GuptNagaRavi2017, KambNagaNavi2017}), we get the following corollary.
\begin{corollary}\label{cor:odt_logm}
    There exists an efficient algorithm that is $O(\log m)$-approximate
    for \pbText{} and with or without unit-cost boxes.
\end{corollary}

%



\subsection{Constant approximation for Partially Adaptive $\pbv$}
Moving on, we show how our reduction can be used to obtain and improve the results of~\cite{ChawGergTengTzamZhan2020}. Recall that in ~\cite{ChawGergTengTzamZhan2020} the authors presented a constant factor approximation algorithm against a Partially Adaptive benchmark where the order of opening boxes must be fixed up front.

In such a case, the reduction of Section~\ref{sec:pb_to_mssc} can be used to reduce $\pbv$ to the standard \textsc{Min Sum Set Cover} (i.e. without feedback), which admits a 4-approximation~\cite{FeigLovaTeta2004}.

\begin{corollary}\label{cor:focs_pa}
    There exists a polynomial time algorithm for $\pbv$ that is $O(1)$-competitive against the partially adaptive benchmark.
\end{corollary}

The same result applies even in the case of non-uniform opening costs. This is because a 4-approximate algorithm for \textsc{Min Sum Set Cover} is known even when elements have arbitrary costs~\cite{MunaBabuMotwWido2005}. The case of non-uniform opening costs has also been considered for \pbText{} by \cite{ChawGergTengTzamZhan2020} but only provide an algorithm to handle polynomially bounded opening costs.



\section{Connecting \pbText{} and $\msscf$}\label{sec:pb_to_mssc}

In this section we establish the connection between  \pbText{} and \msscText{}. We show that the two problems are equivalent up to logarithmic factors in approximation ratio. 

One direction of this equivalence is easy to see in fact: \msscText{} is a special case of \pbText{}. Note that in both problems we examine boxes/elements in an adaptive order. In $\pb$ we stop when we find a sufficiently small value; in $\msscf$ we stop when we find an element that belongs to the instantiated scenario. To establish a formal connection,  given an instance of $\msscf$, we can define the "value" of each element $i$ in scenario $s$ as being $0$ if the element belongs to the set $s$ and as being $L+f_i(s)$ for some sufficiently large value $L$ where $f_i(s)$ is the feedback of element $i$ for set $s$. This places the instance within the framework of $\pb$ and a $\pb$ algorithm can be used to solve it. We formally describe this reduction in Section~\ref{sec:apn_tool} of the Appendix. 

\begin{restatable}{claim}{msscfToPb}\label{cl:msscf_to_pb}
	If there exists an $\alpha(n,m)$-approximation algorithm for $\pb$ then there exists a $\alpha(n,m)$-approximation for $\msscf$.
\end{restatable}

The more interesting direction is a reduction from $\pb$ to $\msscf$. In fact we show that a general instance of $\pb$ can be reduced to the simpler {\em uniform} version of \msscText{}. We devote the rest of this section to proving the following theorem.

\begin{theorem}[\pbText{} to $\msscf$] \label{thm:pb_to_mssc}
If there exists an $a(n,m)$ approximation algorithm for $\umsscf$ then there exists a $O(\alpha(n+m, m^2) \log \alpha(n+m, m^2))$-approximation for $\pb$.
\end{theorem}

\subsubsection*{Guessing a stopping rule and an intermediate problem}
The feedback structure in $\pb$ and $\msscf$ is quite similar, and the main component in which the two problems differ is the stopping condition. In $\msscf$, an algorithm can stop examining elements as soon as it finds one that "covers" the realized set. In $\pb$, when the algorithm observes a value in a box, it is not immediately apparent whether the value is small enough to stop or whether the algorithm should probe further, especially if the scenario is not fully identified. The key to relating the two problems is to "guess" an appropriate stopping condition for $\pb$, namely an appropriate threshold $T$ such that as soon as the algorithm observes a value smaller than this threshold, it stops. We say that the realized scenario is "covered".

To formalize this approach, we introduce an intermediate problem called \emph{\pbText{} with costly outside option $T$} (also called \emph{threshold}), denoted by $\pbT$. In this version the objective is to minimize the cost
of finding a value $\leq T$, while we have the extra option to quit searching
by opening an \emph{outside option} box of cost $T$. We say that a scenario is \emph{covered} in a
given run of the algorithm if it does not choose the outside option box $T$.

We show that \pbText{} can be reduced to $\pbT$ with a logarithmic loss in approximation factor, and then $\pbT$ can be reduced to \msscText{} with a constant factor loss. The following two results capture the details of these reductions.

\begin{restatable}{claim}{upbTToUmsscf}\label{thm:upbT_to_umsscf}
If there exists an $\alpha(n,m)$ approximation algorithm for $\umsscf$ then there exists an $3\alpha(n+m,m^2)$-approximation for $\upbT$.
\end{restatable}

It is also worth noting that $\pbT$ is a special case of the Adaptive Ranking problem which directly implies a $\log m$ approximation factor (given in \cite{KambNagaNavi2017}).

\begin{restatable}{mainLemma}{pbvToUpbLog}\label{thm:pbv_to_upb0_log}
	Given a polynomial-time $\alpha(n,m)$-approximation algorithm for $\upbT$,
	there exists a polynomial-time $ {O}(\alpha(n,m) \log \alpha(n,m))$-approximation for $\pbv$.
\end{restatable}

The relationship between
$\pbT$ and \msscText{} is relatively straightforward and requires explicitly relating the structure of feedback in the two problems. We describe the details in Section~\ref{sec:apn_tool} of the Appendix.

\paragraph{Putting it all together.}
The proof of Theorem~\ref{thm:pb_to_mssc} follows by combining Claim~\ref{thm:upbT_to_umsscf} with Lemmas~\ref{thm:pbv_to_pb0_log} and~\ref{thm:pbv_to_upb0_log} presented in the following sections. Proofs of Claims~\ref{cl:msscf_to_pb}, \ref{thm:upbT_to_umsscf} deferred to Section~\ref{sec:apn_tool} of the Appendix. The rest of this section is devoted to proving Lemmas~\ref{thm:pbv_to_pb0_log} and~\ref{thm:pbv_to_upb0_log}. The landscape of reductions shown in this section is  presented in Figure~\ref{fig:pb_and_mssc}.

\begin{figure}[H]
    \centering
    \newcommand{\bl}[1]{\textcolor{black}{#1}}
\pgfmathsetmacro{\dist}{3}
\pgfmathsetmacro{\distV}{2.5}
\pgfmathsetmacro{\offset}{1.5}
\pgfmathsetmacro{\ypomnimaX}{6}
\pgfmathsetmacro{\ypomnimaY}{1.8}
\begin{tikzpicture}

	\node (pbv) at (0,0){{\Large $\pbv$}};

	\node (upbT) at (\offset,\distV){{\Large $\upbT$}};
 	\node (pbT) at (-\offset,\distV){{\Large $\pbT$}};
  
	\node (umsscf) at (1.5*\offset,2*\distV){{\Large $\umsscf$}};
 
 	\node (msscf) at (-1.5*\offset,2*\distV){{\Large $\msscf$}};

\draw[->,dotted, gray, thick, opacity=0.6] (umsscf) -- (msscf);
   \draw[->,dotted, gray, thick, opacity=0.6] (upbT) --  (pbT);

\draw[->,very thick, dashed] (pbv) -- node [above, left] {\bl{\textbf{Lem}~}\ref{thm:pbv_to_pb0_log}} (pbT);
    \draw[->,very thick, dashed] (pbv) -- node [above, right] {\bl{\textbf{Lem}~}\ref{thm:pbv_to_upb0_log}}(upbT);

 \draw[->, thick] (pbT) -- node [above, left] {Claim~\ref{thm:upbT_to_umsscf}} (msscf);

        \draw[->,thick] (upbT) -- node [above, right] { Claim~\ref{thm:upbT_to_umsscf}}(umsscf);

    \draw[->,dashed, gray] (msscf) to [out=180, in=180]  node [above, left] {Claim~\ref{cl:msscf_to_pb}}  (pbv);

    \draw[->,very thick, dashed] (\ypomnimaX-1, \ypomnimaY+0.8) -- node[above] {{\footnotesize Main Lemma ($\log$ factors)}} (\ypomnimaX+1, \ypomnimaY+0.8);
    
        \draw[->,thick] (\ypomnimaX-1, \ypomnimaY) -- node[above] {{\footnotesize Claim (const. factors)}} (\ypomnimaX+1, \ypomnimaY);
        
    \draw[->,dashed,thick, gray] (\ypomnimaX-1, \ypomnimaY-0.8) -- node[above] {{\footnotesize Minor Claim}} (\ypomnimaX+1, \ypomnimaY-0.8);
    
	  \draw[->,dotted, gray, thick, opacity=0.6] (\ypomnimaX-1, \ypomnimaY-1.6) -- node[above] {{\footnotesize Subproblem}} (\ypomnimaX+1, \ypomnimaY-1.6);


\end{tikzpicture}
    \caption{Reductions shown in this section. Claim~\ref{thm:upbT_to_umsscf} alongside Lemmas~\ref{thm:pbv_to_pb0_log} and~\ref{thm:pbv_to_upb0_log} are part of Theorem~\ref{thm:pb_to_mssc}.}
    \label{fig:pb_and_mssc}
\end{figure}
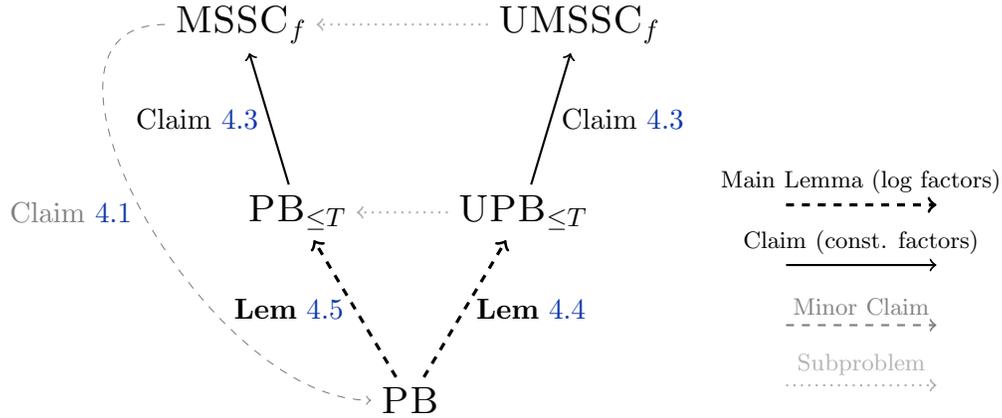

\subsection{Reducing \pbText{} to $\pbT$}\label{subsec:pbt}

Recall that a solution to \pbText{} involves two components ; (1) the order in which to open boxes and (2) a stopping rule.
The goal of the reduction to $\pbT$ is to simplify the stopping rule of the problem, by making values either $0$ or $\infty$, therefore allowing us to focus on the order in which boxes are opened, rather than which value to stop at. 
We start by presenting our main tool, a reduction to \msscText{} in Section~\ref{subsec:main_tool} and then improve upon that to reduce from the \textbf{uniform} version of $\msscf$ (Section~\ref{subsec:improved_tool}).

\subsubsection{Main Tool}\label{subsec:main_tool}
  The high level idea in this reduction is that we repeatedly run the algorithm for $\pbT$ with increasingly larger value of $T$ with the goal of covering some mass of scenarios at every step. The thresholds for every run have to be cleverly
chosen to guarantee that enough mass is covered at every run.  The
distributions on the boxes remain the same, and this reduction does not
increase the number of boxes, therefore avoiding the issues faced by the naive
reduction given in Section~\ref{subsec:naive} of the Appendix. Formally, we show the following lemma.

\begin{mainLemma}\label{thm:pbv_to_pb0_log}
	Given a polynomial-time $\alpha(n,m)$-approximation algorithm for $\pbT$,
	there exists a polynomial-time $ {O}(\alpha(n,m) \log \alpha(n,m))$-approximation for $\pbv$.
\end{mainLemma}

	\begin{algorithm}[H]

\KwIn{Oracle $\mathcal{A}(T)$ for $\pbT$, set of all scenarios $\scenarios$.}
		$i \leftarrow 0$ \tcp{Number of current Phase}
		\While{$\scenarios \neq \emptyset$}{
	Use $\mathcal{A}$ to find smallest $T_i$ via Binary Search s.t. $\Pr{}{\text{accepting the outside option }T_i} \leq 0.2$ \\
	Call the oracle  $\mathcal{A}(T_i)$ on set $\scenarios$ to obtain policy $\pi_i$\\
		    $\scenarios \leftarrow \scenarios \setminus $ \{scenarios with total cost $\leq T_i$\}
		    }
            \For{$i\leftarrow 0$ to $\infty$}{
            Run policy $\pi_i$ until it terminates and selects a box, or accumulates probing cost $T_i$.
            }
      
		\caption{Reduction from $\pbv$ to $\pbT$.}\label{algo:pbv_pb0}
	\end{algorithm}
\mbox{}\\

We will now analyze the policy produced by this algorithm. 

\begin{proof}[Proof of Main Lemma~\ref{thm:pbv_to_pb0_log}]
We start with some notation.  Given an instance $\mathcal{I}$ of $\pbv$, we repeatedly run $\pbT$ in \emph{phases}. Phase $i$ consists of running $\pbT$ with threshold $T_i$ on a sub instance of the original problem where we are left with a smaller set of scenarios, with their probabilities reweighted to sum to $1$. Call this set of scenarios $\scenarios_i$ for phase $i$ and the corresponding instance $\mathcal{I}_i$. 
  After every
	phase $i$, we remove the probability mass that was covered\footnote{Recall, a
	scenario is \emph{covered} if it does not choose the outside option box.},
	and run $\pbT$ on this new instance with a new threshold $T_{i+1}$. In each
	phase, the boxes, costs and values remain the same, but the stopping condition changes: thresholds $T_i$ increase in every subsequent phase. The
	thresholds are chosen such that at the end of each phase, $0.8$ of the
	remaining probability mass is covered. The reduction process is formally
	shown in Algorithm~\ref{algo:pbv_pb0}.

	\paragraph{Accounting for the cost of the policy.} We first note that the total cost of the policy in phase $i$ conditioned on reaching that phase is at most $2T_i$: if the policy terminates in that phase, it selects a box with value at most $T_i$. Furthermore, the policy incurs probing cost at most $T_i$ in the phase. We can therefore bound the total cost of the policy as $\le 2\sum_{i=0}^\infty (0.2)^i T_i$.\\

\noindent
    We will now relate the thresholds $T_i$ to the cost of the optimal PB policy for $\mathcal{I}$. To this end, we define corresponding thresholds for the optimal policy that we call \emph{$p$-thresholds}. Let $\pi^*_\mathcal{I}$ denote the optimal PB policy for $\mathcal{I}$ and let $c_s$ denote the cost incurred by $\pi^*_\mathcal{I}$ when scenario $i$ is realized. A $p$-threshold is the minimum possible threshold $T$ such that
at most $p$ mass of the scenarios has cost more than $T$ in $\pbv$, formally defined below.
\begin{definition}[$p$-Threshold]
	Let $\mathcal{I}$ be an instance of $\pbv$ and $c_s$ be the cost of
	scenario $s\in \scenarios$ in $\pi^*_\mathcal{I}$, we define the $p$-threshold as
	\[
		t_p = \min\{T : \Pr{}{c_s > T} \leq p\}.
	\]
\end{definition}

The following two lemmas relate the cost of the optimal policy to the $p$-thresholds, and the $p$-thresholds to the thresholds $T_i$ our algorithm finds. The proofs of both lemmas are deferred to
Section~\ref{subsec:apn_main_tool} of the Appendix. 
We first formally define a
\emph{sub-instance} of the given \pbText{} instance.

\begin{definition}[Sub-instance]
		Let $\mathcal{I}$ be an instance of $\{ \pbT, \pbv \}$ with set of scenarios
	$\mathcal{S}_\mathcal{I}$ each with probability $p^\mathcal{I}_s$. For any
	$q\in [0,1]$ we call $\mathcal{I}'$ a $q$-sub instance of $\mathcal{I}$ if
	$\mathcal{S}_{\mathcal{I}'} \subseteq \mathcal{S}_\mathcal{I}$ and
	$\sum_{s\in \mathcal{S}_{\mathcal{I}'}} p_s^\mathcal{I} = q$.
\end{definition}

\begin{restatable}{lemma}{pbvOpt}\label{lem:pbv_opt}(Optimal Lower Bound) Let $\mathcal{I}$ be
	the instance of $\pbv$. For any $q<1$, any $\alpha>1$, and $\beta \geq 2$, for the optimal policy $\pi^*_{\mathcal{I}}$ for $\pbv$ it that
	\[
		\cost(\pi^*_\mathcal{I}) \geq \sum_{i=0}^\infty \frac{1}{\beta\alpha} \cdot \lp(q\rp)^{i} t_{q^i/\beta\alpha}.
	\]
\end{restatable}

\begin{restatable}{lemma}{thresholdBound}\label{lem:thresh_bound}
	Given an instance $\mathcal{I}$ of $\pbv$; an $\alpha$-approximation
	algorithm $\mathcal{A}_T$ to $\pbT$; and any $q<1$ and $\beta\ge 2$, suppose that the threshold $T$ satisfies  
	\[
		T \geq t_{q/(\beta \alpha)} 
		+ \beta  \alpha \sum_{\substack{ c_s \in [t_q, t_{q/(\beta \alpha)}]\\ s\in \scenarios }} c_s \frac{p_s}{q}. 		
	\]
Then if $\mathcal{A}_T$ is run on a $q$-sub instance 
of $\mathcal{I}$ with threshold $T$, at most a total mass of $(2/\beta)q$ of the scenarios pick the outside option box $T$.
\end{restatable}

	\paragraph{Calculating the thresholds.}
	For every phase $i$ we choose a threshold $T_i$ such that $T_i = \min \{T :
	\Pr{}{c_s>T} \leq 0.2\}$ i.e. at most $0.2$ of the probability mass of
	the scenarios are not covered. In order to
	select this threshold, we do binary search starting from $T = 1$, running
	every time the $\alpha$-approximation algorithm for $\pbT$ with outside
	option box $T$ and checking how many scenarios select it.  We denote by
	$\intv_i = [t_{(0.2)^i}, t_{(0.2)^i/(10\alpha)} ] $ the relevant interval of costs at every run of the algorithm, then by
	Lemma~\ref{lem:thresh_bound} for $\beta=10$, we know that for remaining total probability
	mass $(0.2)^i$, any threshold which satisfies
	\begin{equation*}
		T_i \geq t_{(0.2)^{i-1}/10a} 
			+ 10 \alpha \sum_{\substack{s\in \scenarios\\ c_s \in \intv_i }} c_s \frac{p_s}{(0.2)^i} 
	\end{equation*}
	also satisfies the desired covering property, i.e. at least $0.8$ mass of the
	current scenarios is covered. Therefore the threshold

 $T_i$ found by our binary search satisfies the following 
 \begin{equation}\label{eq:threshold}
		T_i = t_{(0.2)^{i-1}/10a} 
			+ 10 \alpha \sum_{\substack{s\in \scenarios\\ c_s \in \intv_i }} c_s \frac{p_s}{(0.2)^i}. 
	\end{equation}

	\paragraph{Bounding the final cost.}
	To bound the final cost, we recall that at the end of every phase we cover $0.8$ of the remaining
	scenarios. Furthermore, we observe that each threshold $T_i$ is charged in the above Equation \eqref{eq:threshold} to optimal costs of scenarios corresponding to intervals of the
	form $ \intv_i =[t_{(0.2)^i}, t_{(0.2)^i/(10\alpha)} ]$. Note
	that these intervals are overlapping. We therefore get
	\begin{align*}
			\cost(\pi_\mathcal{I}) &\leq  2 \sum_{i=0}^\infty (0.2)^i T_i & \\
				 & = 2 \sum_{i=0}^\infty \lp( (0.2)^i t_{(0.2)^{i-1}/10a} 
				 + 10 \alpha \sum_{\substack{s\in \scenarios\\ c_s \in \intv_i }} c_s p_s \rp)
			 		& \text{From equation \eqref{eq:threshold}} \\
			& \leq 4\cdot 10\alpha \pi^*_\mathcal{I} + 
				20 \alpha \sum_{i=0}^\infty  \sum_{\substack{s\in \scenarios\\ c_s \in \intv_i }} c_s p_s 
					&\text{Using Lemma~\ref{lem:pbv_opt} for $\beta=10, q=0.2$}\\
			&\leq 40\alpha \log \alpha \cdot \pi^*_\mathcal{I}. & 
	\end{align*}
	Where the last inequality follows since each scenario with cost $c_s$ can
	belong to at most $\log \alpha$ intervals, therefore we get the theorem.
\end{proof}
 
	Notice the generality of this reduction; the distributions on the values
	are preserved, and we did not make any more assumptions on the scenarios or
	values throughout the proof. Therefore we can apply this tool regardless of
	the type of correlation or the way it is given to us, e.g. we could be
	given a parametric distribution, or an explicitly given distribution, as we
	see in the next section.
		
\subsubsection{An Even Stronger Tool}\label{subsec:improved_tool}
Moving one step further, we show that if we instead of $\pbT$ we had an $\alpha$-approximation algorithm for $\upbT$ we can obtain the same guarantees as the ones described in Lemma~\ref{thm:pbv_to_pb0_log}. Observe that we cannot directly use  Algorithm~\ref{algo:pbv_pb0} since the oracle now requires that all scenarios have the same probability, while this might not be the case in the initial $\pbv$ instance. The theorem stated formally follows.

\pbvToUpbLog*

We are going to highlight the differences with the proof of Main Lemma~\ref{thm:pbv_to_pb0_log}, and show how to change Algorithm~\ref{algo:pbv_pb0} to work with the new oracle, that requires the scenarios to have uniform probability. The function \textbf{Expand} shown in Algorithm~\ref{algo:pbv_upb0} is used to transform the instance of scenarios to a uniform one where every scenario has the same probability by creating multiple copies of the more likely scenarios. The function is formally described in Algorithm~\ref{algo:expand} in Section~\ref{sec:apn_new_tool} of the Appendix, alongside the proof of Main Lemma~\ref{thm:pbv_to_upb0_log}.\\

	\begin{algorithm}[H]
	\KwIn{Oracle $\mathcal{A}(T)$ for $\upbT$, set of all scenarios $\scenarios$, $c=1/10, \delta=0.1$.}
		$i \leftarrow 0$ \tcp{Number of current Phase}
		\While{$\scenarios \neq \emptyset$}{
		Let $\mathcal{L} = \lp\{s\in \scenarios : p_s \leq c\cdot \frac{1}{|\scenarios|}\rp\}$ \tcp{Remove low probability scenarios}
		$\scenarios' = \scenarios \setminus \mathcal{L}$\\
		$\mathcal{UI} = $ Expand($\scenarios'$)\\
	In instance $\mathcal{UI}$ use $\mathcal{A}$ to find smallest $T_i$ via Binary Search s.t. $\Pr{}{\text{accepting }T_i} \leq \delta$ \\
	Call the oracle  $\mathcal{A}(T_i)$\\
    $\scenarios \leftarrow \big( \scenarios' \setminus  \{ s\in\scenarios': c_s\leq T_i \} \big) \cup \mathcal{L}$
		    }
		\caption{Reduction from $\pbv$ to $\upbT$.}\label{algo:pbv_upb0}
	\end{algorithm}

\section{Connecting $\msscf$ and \dtText}\label{sec:mssc_to_dt}
In this section we establish the connection between \msscText{} and \dtText{}. We show that the uniform versions of these problems are equivalent up to constant factors in approximation ratio. The results of this section are summarized in Figure~\ref{fig:mssc_and_odt} and the two results below.

\begin{figure}[H]
    \centering
    \newcommand{\bl}[1]{\textcolor{black}{#1}}
\pgfmathsetmacro{\dist}{3}
\pgfmathsetmacro{\distV}{1.9}
\pgfmathsetmacro{\offset}{1}
\pgfmathsetmacro{\ypomnimaX}{6}
\pgfmathsetmacro{\ypomnimaY}{7}
\begin{tikzpicture}

	\node (umsscf) at (2*\offset,2*\distV){{\large $\umsscf$}};
 	\node (msscf) at (-2*\offset,2*\distV){{\large $\msscf$}};

	\node (udt) at (2*\offset, 3*\distV){{\large $\udt$}};
 	\node (odt) at (-2*\offset, 4*\distV){{\large $\odt$}};
	\draw[->,dotted, gray, thick, opacity=0.6] (umsscf) -- (msscf);
 
\draw[->,dotted, gray, thick, opacity=0.6] (udt) --  (odt);

 
    \draw[->,dashed, thick, gray] (msscf) -- node [above, left] {\bl{Claim~}\ref{thm:msscf_to_odt}} (odt); 
    
    \draw[->, dashed, thick, gray] (umsscf) to [out=135,in=180]  node [above, left] {\textcolor{black}{Claim~\ref{thm:msscf_to_odt}}}  (udt);

    \draw[->,very thick] (udt) to [in=45, out=0] node [above, right] {\bl{Thm~}\ref{thm:udt_to_pbT}}  (umsscf);

    
        \draw[->,thick] (\ypomnimaX-1, \ypomnimaY) -- node[above] {{\footnotesize Main Theorem: const. factors}} (\ypomnimaX+1, \ypomnimaY);
        
    \draw[->,dashed,thick, gray] (\ypomnimaX-1, \ypomnimaY-0.8) -- node[above] {{\footnotesize Minor Claim}} (\ypomnimaX+1, \ypomnimaY-0.8);
    
	  \draw[->,dotted, gray, thick, opacity=0.6] (\ypomnimaX-1, \ypomnimaY-1.6) -- node[above] {{\footnotesize Subproblem}} (\ypomnimaX+1, \ypomnimaY-1.6);

\end{tikzpicture}
    \caption{Summary of reductions in Section~\ref{sec:mssc_to_dt}}
    \label{fig:mssc_and_odt}
\end{figure}
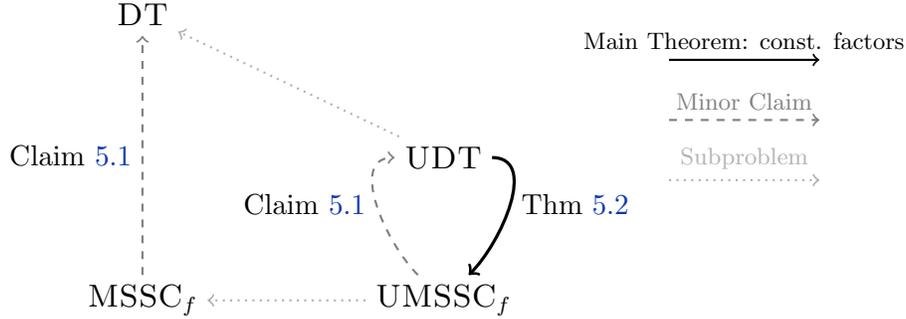

\begin{restatable}{claim}{msscfToOdt}\label{thm:msscf_to_odt}
	If there exists an $\alpha(n,m)$-approximation algorithm for $\odt$ ($\udt$) then there exists a $\lp( 1+\alpha(n,m)\rp)$-approximation algorithm for $\msscf$ (resp. $\umsscf$).
\end{restatable}

\begin{restatable}[\udtText{} to $\umsscf$]{theorem}{udtTopbT}\label{thm:udt_to_pbT}
Given an $\alpha(m,n)$-approximation algorithm for  $\umsscf$ then there exists an
	$O(\alpha(n+m, m))$-approximation algorithm for $\udt$.
\end{restatable}

The formal proofs of these statements can be found in Section~\ref{sec:apn_mssc_and_odt} of the Appendix. Here we sketch the main ideas. 

One direction of this equivalence is again easy to see. The main difference between \dtText{}  and $\msscf$ is that the former requires scenarios to be exactly identified whereas in the latter it suffices to simply find an element that covers the scenario. In particular, in $\msscf$ an algorithm could cover a scenario without identifying it by, for example, covering it with an element that covers multiple scenarios. To reduce $\msscf$ to $\odt$ we simply 
introduce extra feedback into all of the elements of the $\msscf$ instance such that the elements isolate any scenarios they cover. (That is, if the algorithm picks an element that covers some subset of scenarios, this element provides feedback about which of the covered scenarios materialized.) This allows us to relate the cost of isolation and the cost of covering to within the cost of a single additional test, implying Claim~\ref{thm:msscf_to_odt}.


\paragraph{Proof Sketch of Theorem~\ref{thm:udt_to_pbT}.} The other direction is more complicated, as we want to ensure that covering implies isolation. Given an instance of $\udt$, we create a special element for each scenario which is the unique element covering the scenario and also isolates the scenario from all other scenarios. The intention is that an algorithm for $\msscf$ on this new instance only chooses the special isolating element in a scenario after it has identified the scenario. If that happens, then the algorithm's policy is a feasible solution to the $\udt$ instance and incurs no extra cost. The problem is that an algorithm for $\msscf$ over the modified instance may use the special covering element before isolating a scenario. We argue that this choice can be "postponed" in the policy to a point at which isolation is nearly achieved without incurring too much extra cost. This involves careful analysis of the policy's decision tree and we present details in the appendix.

\paragraph{Why our reduction does not work for $\odt$.} Our analysis above heavily uses the fact that the probabilities of all scenarios in the $\udt$ instance are equal. This is because the "postponement" of elements charges increased costs of some scenarios to costs of other scenarios. In fact, our reduction above fails in the case of non-uniform distributions over scenarios -- it can generate an $\msscf$ instance with optimal cost much smaller than that of the original $\odt$ instance. 



To see this, consider an example with $m$ scenarios where scenarios $1$ through $m-1$ happen with probability $\e/(m-1)$ and scenario $m$ happens with probability $1-\e$. There are $m-1$ tests of cost $1$ each. Test $i$ for $i\in [m-1]$ isolates scenario $i$ from all others. Observe that the optimal cost of this $\odt$ instance is at least $(1-\e)(m-1)$ as all $m-1$ tests need to be run to isolate scenario $m$. Our construction of the $\msscf$ instance adds another isolating test for scenario $m$. A solution to this instance can use this new test at the beginning to identify scenario $m$ and then run other tests with the remaining $\e$ probability. As a result, it incurs cost at most $(1-\e)+\e(m-1)$, which is a factor of $1/\e$ cheaper than that of the original $\odt$ instance.



\section{Mixture of Product Distributions}\label{sec:mixt}
	In this section we switch gears and consider the case where we are given a mixture of $m$ product distributions.
Observe that using the tool described in
Section~\ref{subsec:main_tool}, we can reduce this problem to $\pbT$. This now
is equivalent to the noisy version of $\odt$
\cite{GoloKrau2017correct,JiaNagaNaviRavi2019} where for a specific scenario,
the result of each test is not deterministic and can get different values with
different probabilities.

\paragraph{Comparison with previous work:} previous work on noisy decision tree, considers limited noise models or the
runtime and approximation ratio depends on the type of noise. For example in
the main result of \cite{JiaNagaNaviRavi2019}, the noise outcomes are binary
with equal probability. The authors mention that it is possible to extend the
following ways:
\begin{itemize}
	\item to probabilities within $[\delta, 1-\delta]$, incurring an extra $
			1/\delta$ factor in the approximation
	\item to non-binary noise outcomes, incurring an extra at most $m$ factor
			in the approximation
\end{itemize}
Additionally, their algorithm works by expanding the scenarios for every
possible noise outcome (e.g. to $2^m$ for binary noise). In our work the number
of noisy outcomes does not affect the number of scenarios whatsoever.

In our work, we obtain a \textbf{constant approximation} factor, that does not
depend in any way on the type of the noise. Additionally, the outcomes of the
noisy tests can be arbitrary, and do not affect either the approximation factor
or the runtime. We only require a \emph{separability} condition to hold ; the
distributions either differ \emph{enough} or are exactly the same.  Formally,
we require that for any two scenarios $s_1, s_2\in \scenarios$ and for every
box $i$, the distributions $\mathcal{D}_{is_1}$ and $\mathcal{D}_{is_2}$
satisfy $\lp| \mathcal{D}_{is_1} - \mathcal{D}_{is_2} \rp| \in \mathbb{R}_{\geq
\e} \cup \{0\}$, where $|\mathcal{A} - \mathcal{B}|$ is the total variation
distance of distributions $\mathcal{A}$ and $\mathcal{B}$.

\subsection{A DP Algorithm for noisy $\pbT$}
	We move on to designing a dynamic programming algorithm to solve the
	$\pbT$ problem, in the case of a mixtures of product distributions. The
	guarantees of our dynamic programming algorithm are given in the following
	theorem.

	\begin{theorem}\label{thm:dp}
		For any $\beta >0$, let $\pi_{\scaleto{DP}{5pt}}$ and $\pi^*$ be the
		policies produced by Algorithm $\DP(\beta)$ described by
		Equation~\eqref{eq:dp} and the optimal policy 
		respectively and $\ub = \frac{m^2}{\e^2}
		\log \frac{m^2T}{c_{\min}\beta}$. Then it holds that 
		\[
			c(\pi_{\scaleto{\DP}{5pt}}) \leq (1+\beta)c(\pi^*)
		.\] 
		and the $\DP$ runs in time $n^{\ub}$, where
		$n$ is the number of boxes and $c_{\min}$ is the minimum cost box.
	\end{theorem}

	Using the reduction described in Section~\ref{subsec:main_tool} and the
	previous theorem we can get a constant-approximation algorithm for the
	initial $\pbv$ problem given a mixture of product distributions. 	Observe that in the
	reduction, for every instance of $\pbT$ it runs, the chosen threshold $T$
	satisfies that $T \leq (\beta+1) c(\pi^*_T)/0.2$ where $\pi^*_T$ is the
	optimal policy for the threshold $T$. The inequality holds since the
	algorithm for the threshold $T$ is a $(\beta+1)$ approximation and it
	covers $80\%$ of the scenarios left (i.e. pays $0.2T$ for the rest).
	This is formalized in the following corollary.

	\begin{corollary}\label{cor:mixt}
		Given an instance of $\pbv$ on $m$ scenarios, and the DP algorithm
		described in Equation~\eqref{eq:dp}, then using Algorithm~\ref{algo:pbv_pb0}
		 we obtain an $O(1)$-approximation
		algorithm for $\pbv$ that runs in $n^{\tilde{O}(m^2/\e^2)}$.
	\end{corollary}
	
	Observe that the naive DP, that keeps track of all the boxes and possible
	outcomes, has space exponential in the number of boxes, which can be very
	large. In our DP, we exploit the separability property of the distributions
	by distinguishing the boxes in two different types based on a given set of
	scenarios. Informally, the \emph{informative} boxes help us distinguish
	between two scenarios, by giving us enough TV distance, while the
	\emph{non-informative} always have zero TV distance. The formal definition
	follows.
	
	\begin{definition}[Informative and non-informative boxes]
			Let $S\subseteq \scenarios$ be a set of scenarios. Then we call a box $k$ \emph{informative} if
		there exist $s_i, s_j\in \scenarios$ such that
		\[
			| \mathcal{D}_{ks_i} - \mathcal{D}_{ks_j}| \geq \e.
		\]
		We denote the set of all \emph{informative} boxes by $\ib(S)$.
		Similarly, the boxes for which the above does not hold are
		called \emph{non-informative} and the set of these boxes is denoted by
		$\nib(S)$.
	\end{definition}
	
	\paragraph{Recursive calls of the DP:}Our dynamic program chooses
	at every step one of the following options:
	\begin{enumerate}
		\item open an \textbf{informative} box: this step contributes towards
				\emph{eliminating} improbable scenarios. From the definition of
				informative boxes, every time such a box is opened, it gives
				TV distance at least $\e$ between at least two
				scenarios, making one of them more probable than the other. We
				show
				(Lemma~\ref{lem:dp_tests}) that it takes a finite amount of
				these boxes to decide, with high probability, which scenario is
				the one realized (i.e. eliminating all but one scenarios). 
		\item open a \textbf{non-informative} box: this is a greedy step; the
				best non-informative box to open next is the one that maximizes
				the probability of finding a value smaller than $T$. Given a
				set $S$ of scenarios that are not yet eliminated, there is a unique next
					non-informative box which is best.  We denote by $\nib^*(S)$ the
				function that returns this next best non-informative box.
				Observe that the non-informative boxes do not affect the greedy
				ordering of which is the next best, since they do not affect
				which scenarios are eliminated.
	\end{enumerate}

	 \paragraph{State space of the DP:} the DP keeps track of the following three quantities:
	 \begin{enumerate}
		 \item \textbf{a list $M$} which  consists of sets of informative boxes
				 opened and numbers of non-informative ones opened in between
					 the sets of informative ones.  Specifically, $M$ has the
					 following form: $M = S_1| x_1 |S_2| x_2| \ldots |S_L
					 |x_L$\footnote{If $b_i$ for $i\in [n]$ are boxes, the list
					 $M$ looks like this: $b_3 b_6 b_{13}| 5 | b_{42} b_1 | 6 |
					 b_2$} where $S_i$ is a set of informative boxes, and
					 $x_i\in \mathbb{N}$ is the number of non-informative boxes
					 opened exactly after the boxes in set $S_i$. We also
					 denote by $\ib(M)$ the informative boxes in the list $M$.

			In order to update $M$ at every recursive call, we either append a
					 new informative box $b_i$ opened (denoted by $M| b_i$) or,
					 when a non-informative box is opened, we add $1$ at the end, denoted by $M+1$.
			 \item \textbf{a list $E$} of $m^2$ tuples of integers $(z_{ij},
					 t_{ij})$, one for each pair of distinct scenarios $(s_i,
					 s_j)$ with $i,j\in[m]$. The number $z_{ij}$ keeps track of
					 the number of informative boxes between $s_i$ and $s_j$
					 that the value discovered had higher probability for
					 scenario $s_i$,
					 and the number $t_{ij}$ is the total number of informative
					 for scenarios $s_i$ and $s_j$ opened. Every time an
					 informative box is opened, we increase the $t_{ij}$
					 variables for the scenarios the box was informative
					 and add $1$ to the $z_{ij}$ if the value discovered
					 had higher probability in $s_i$.  When a non-informative
					 box is opened, the list remains the same.We denote this
					 update by $E^{\scaleto{++}{5pt}}$.
			 \item \textbf{a list $S$} of the scenarios not yet eliminated. Every time
					 an informative test is performed, and the list $E$
					 updated, if for some scenario $s_i$ there exists another
					 scenario $s_j$ such that $t_{ij} > 1/\e^2 \log (1/\delta)$
					 and $|z_{ij} - \E{}{z_{ij}|s_i}| \leq \e/2$ then $s_j$ is
					 removed from $S$, otherwise $s_i$ is removed\footnote{This is the process of elimination in the proof of Lemma~\ref{lem:dp_tests}}. This update is denoted by $S^{\scaleto{++}{5pt}}$.
	 \end{enumerate}

	 \paragraph{Base cases:} 
	 if a value below $T$ is found, the algorithm stops. The
	 other base case is when $|S|=1$, which means that the scenario realized is identified, we either take the
	 outside option $T$ or search the boxes for a value below $T$, whichever is
	 cheapest. If the scenario is identified correctly, the DP finds the expected
	 optimal for this scenario. We later show that we make a mistake only with  low
	 probability, thus increasing the cost only by a constant factor.
	 We denote by $\nat(\cdot,
	 \cdot, \cdot)$ the ``nature's" move, where the value in the box we chose
	 is realized, and $\sol(\cdot, \cdot, \cdot)$ is the minimum value obtained
	 by opening boxes. The recursive formula is shown below.
	 \newcommand{\plus}{\raisebox{.3\height}{\scalebox{.7}{+}}}
	\begin{align}\label{eq:dp}
			\begin{split}
		\sol(M, E, S) & 
				= \begin{cases}
						\min ( T, c_{\scaleto{\nib^*(S)}{6pt}} + \nat(M \plus 1, E, S) ) & \text{ if } |S| = 1\\
						\min\Big(T, \min\limits_{\scaleto{i\in \ib(M)}{6pt}}\lp( c_i\plus \nat(M|i, E, S)\rp)  & \\
							\hspace{1.5cm} ,c_{\scaleto{\nib^*(S)}{6pt}} + \nat(M\plus 1, E, S)  \Big) & \text{else}
				\end{cases}\\
		\nat(M, E, S) & 
				= \begin{cases}
						0  &\hspace{3.35cm}\text{if } v_\text{last box opened} \leq T\\
						\sol(M, E^{\scaleto{++}{4pt}}, S^{\scaleto{++}{4pt}})	& \hspace{3.35cm}\text{else}
			\end{cases}
			\end{split}
	\end{align}
	The final solution is $\DP(\beta) = \sol(\emptyset , E^0, \scenarios)$,
	where $E^0$ is a list of tuples of the form $(0,0)$, and in order to update
	$S$ we set $\delta = \beta c_{\min}/(m^2T)$.

	\begin{restatable}{lemma}{dpTests}\label{lem:dp_tests}
		Let $s_1, s_2\in \scenarios$ be any two scenarios. Then after opening
		$\frac{\log (1/\delta)}{\e^2}$ informative boxes, we can eliminate
		one scenario with probability at least $1-\delta$.
	\end{restatable}

	We defer the proof of this lemma and Theorem~\ref{thm:dp} to Section~\ref{sec:apn_mixt} of the
	Appendix.

\bibliographystyle{alpha}
\bibliography{reference}
\newpage 

\appendix

\section{A Naive Reduction to $\pbT$}\label{subsec:naive}
In this section we present a straightforward reduction from \pbText{} to $\pbT$  as an alternative to Theorem~\ref{thm:pb_to_mssc}. This reduction has a simpler construction compared to the reduction of Section~\ref{sec:pb_to_mssc}, and does not lose a logarithmic factor in the approximation, it however faces the following issues.

\begin{enumerate}
\item It incurs an extra computational cost, since it adds a number of boxes that depends on the size of the values' support.
\item It requires opening costs, which means that the oracle for \pbText{} with outside option should be able to handle non-unit costs on the boxes, even if the original $\pb$ problem had unit-cost boxes.
\end{enumerate}

We denote by $\pbTc$ the version of \pbText{} with outside option that has \textbf{non-unit} cost boxes, and formally state the guarantees of our naive reduction below.

\begin{restatable}{theorem}{pbvToPbZero}\label{thm:pbv_to_pb0}
For $n$ boxes and $m$ scenarios,
given an $\alpha(n,m)$-approximation algorithm for $\pbTc$ for arbitrary $T$, there
exits a $2\alpha(n\cdot |\text{supp}(v)|,m)$-approximation for $\pbv$ that runs in polynomial time in the
number of scenarios, number of boxes, and the number of values.
\end{restatable}

Figure~\ref{fig:tool} summarizes all the reductions from $\pb$ to $\pbT$ and in Table~\ref{table:removing_vals} we compare the properties of the naive reduction of this section, to the ones show in Section~\ref{sec:pb_to_mssc}. The main differences are that there is a blow-up in the number of boxes that depends on the support, while losing only constant factors in the approximation.

\begin{figure}[H]
    \centering
\pgfmathsetmacro{\dist}{3}
\pgfmathsetmacro{\distV}{2}
\pgfmathsetmacro{\offset}{4}
\pgfmathsetmacro{\ypomnimaX}{6.5}
\pgfmathsetmacro{\ypomnimaY}{0}
\begin{tikzpicture}

	\node (upbT) at (-\offset,\distV){{\large $\upbT$}};
 
	\node (pbv) at (0,0){{\large $\pbv$}};

	\node (pbT) at (0,\distV){{\large $\pbT$}};
	\node (pbTc) at (\offset,\distV){{\large $\pbTc$}};
	
	 \draw[->,dotted, gray, thick, opacity=0.6] (pbT) -- (pbTc);
    \draw[<-,dotted, gray, thick, opacity=0.6] (pbT) -- (upbT);

	 \draw [->,very thick, dashed] (pbv) -- node [above, left] {Main Lem. \ref{thm:pbv_to_upb0_log}\,\,} (upbT);

	 \draw [->,very thick, dashed] (pbv) -- node [above] {Main Lem. \ref{thm:pbv_to_pb0_log}} (pbT);
	 	 
	 	\draw [->, thick] (pbv) -- node [above, right, thick] {Thm \ref{thm:pbv_to_pb0}} (pbTc);

     \draw[->,very thick, dashed] (\ypomnimaX-1, \ypomnimaY+0.8) -- node[above] {{\footnotesize Main Lemma ($\log$ factors)}} (\ypomnimaX+1, \ypomnimaY+0.8);
    
    \draw[->,thick] (\ypomnimaX-1, \ypomnimaY) -- node[above] {{\footnotesize Main Theorem (const. factors)}} (\ypomnimaX+1, \ypomnimaY);
        
    
	  \draw[->,dotted, gray, thick, opacity=0.6] (\ypomnimaX-1, \ypomnimaY-0.8) -- node[above] {{\footnotesize Subproblem}} (\ypomnimaX+1, \ypomnimaY-0.8);
	 
\end{tikzpicture}
    \caption{Reductions shown in Section~\ref{subsec:pbt}}
    \label{fig:tool}
\end{figure}
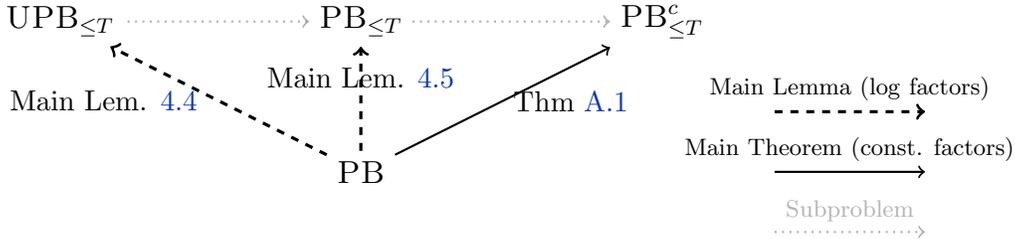

\begin{table}[H]
\centering
\begin{tabular}{ccc}
\textbf{}                                            & \multicolumn{2}{c}{\textbf{Reducing $\pbv$ to}}                                                                                        \\ \cline{2-3} 
\multicolumn{1}{c|}{}                                & \multicolumn{1}{c|}{\textbf{$\pbTc$, Theorem~\ref{thm:pbv_to_pb0}}}                           & \multicolumn{1}{c|}{\textbf{$(\mathcal{U})\pbT$, Main Lemma~\ref{thm:pbv_to_pb0_log} (\ref{thm:pbv_to_upb0_log})}}                           \\ \hline
\multicolumn{1}{|c|}{\textbf{Costs of boxes}}                 & \multicolumn{1}{c|}{ Introduces non-unit costs}        & \multicolumn{1}{c|}{ Maintains costs}        \\ \hline
\multicolumn{1}{|c|}{\textbf{Probabilities}}         & \multicolumn{1}{c|}{ Maintains probabilities} & \multicolumn{1}{c|}{ \begin{tabular}[c]{@{}c@{}}Maintains probabilities  \\ (Makes probabilities uniform)\end{tabular}} \\ \hline
\multicolumn{1}{|c|}{\textbf{\# of extra scenarios}} & \multicolumn{1}{c|}{ 0}                      & \multicolumn{1}{c|}{ 0}                 \\ \hline
\multicolumn{1}{|c|}{\textbf{\# of extra boxes}}     & \multicolumn{1}{c|}{ $n \cdot \supp(v)$}          & \multicolumn{1}{c|}{ 0}                     \\ \hline
\multicolumn{1}{|c|}{\textbf{Approximation loss}} & \multicolumn{1}{c|}{ $2\alpha(n\cdot \supp(v),m)$} & \multicolumn{1}{c|}{ $O(\alpha(n,m) \log a(n,m))$}                       \\ \hline
\end{tabular}
\caption{Differences of reductions of Theorems~\ref{thm:pbv_to_pb0}, and the Main Lemmas
\ref{thm:pbv_to_pb0_log} and \ref{thm:pbv_to_upb0_log} that comprise Theorem~\ref{thm:pb_to_mssc}.}
\label{table:removing_vals}
\end{table}

The main idea is that we can move the information about the values contained in the boxes into
the cost of the boxes. We do achieve this effect by creating one new box for every (box, value)-pair. Note, that doing this risks losing the information about the realized scenario that the original boxes revealed. To retain this information, we keep the original boxes, but replace their values by high values. The high values guarantee the effect of the new boxes is retained. Now, we can formalize this intuition.

\paragraph{$\pbT$ Instance.}
Given an instance $\mathcal{I}$
of
$\pbv$, we construct an instance $\mathcal{I}'$
of $\pbT$. We need $T$ to be sufficiently large so that the outside option is never chosen. The net effect is that a policy for
$\pbv$ is easily inferred from a policy for $\pbT$. We define the instance $\mathcal{I}'$ to have the same scenarios $s_i$ and same scenario probabilities $p_i$ as $\mathcal{I}$. We choose $T=\infty$\footnote{We set $T$ to a value larger than $\sum_i c_i + \max_{i,j} v_{ij}$.}, and define the new values by
$v'_{i,j} = v_{i,j} + T + 1$. Note that all of these values will be larger than $T$
and so a feasible policy cannot terminate after receiving such a value. At the same time, these values ensure the same branching behaviour as before since each distinct value is mapped one to one to a new distinct value. Next, we add additional ``final" boxes for
each pair $(j,v)$ where $j$ is a box and $v$ a potential value of box $j$. Each ``final"
box $(j,v)$ has cost $c_j + v$. Box $(j,v)$ has value $0$ for the scenarios where box $j$
gives exactly value $v$ and values $T+1$ for all other scenarios. Formally,
\begin{align*}
    v'_{i,(j,v)} = \begin{cases}
        0 & \text{if } v_{i,j} = v\\
		T+1 & \text{else}
    \end{cases}
\end{align*}
Intuitively, these ``final'' boxes indicate to a policy that this will be the last box opened, and so its values, which are at least that of the best values of the boxes chosen, should now be taken into account in the cost of the solution. 
 
In order to prove Theorem~\ref{thm:pbv_to_pb0}, we use two key lemmas. In Lemma~\ref{lem:pbv_to_pb0_opts} we show that the optimal value for the transformed instance $\mathcal{I}'$ of $\pbT$ is not much higher than the optimal value for original instance $\mathcal{I}$. In
Lemma~\ref{lem:pbv_to_pb0_policies} we show how to obtain a policy for the
initial instance with values, given a policy for the problem with a threshold.

\begin{lemma}\label{lem:pbv_to_pb0_opts}
	Given the instance $\mathcal{I}$ of $\pbv$ and the constructed instance $\mathcal{I}'$
	of $\pbT$ it holds
	that
	\[
		c(\pi^*_{\mathcal{I}'}) \leq 2 c(\pi^*_{\mathcal{I}}).
	\]
\end{lemma}

\begin{proof}
	We show that given an optimal policy for $\pbv$, we can construct a
	feasible policy $\pi'$ for $\mathcal{I}'$ such that
	$c(\pi_{0}) \leq 2c(\pi^*_{\mathcal{I}})$.  We construct the policy
	$\pi'$ by opening the same boxes as $\pi$ and
	finally opening the corresponding ``values" box, in order to find the $0$ needed to stop.  
	
	Fix any scenario $i$, and suppose box $j$ achieved the smallest value $V_{i,j}$ of all boxes opened under scenario $i$. Since $j$ is opened, in the instance
	$\mathcal{I}'$ we open box $(j,v_{i,j})$, and from the construction of
	$\mathcal{I}'$ we have that $v'_{i,(j,v_{i,j})} = 0$. Since on every
	branch we open a box with values $0$\footnote{$\pi$ opens at least
	one box.}, we see that $\pi'$ is a feasible policy for $\mathcal{I}'$. 
	Under scenario $i$, we have that the cost of $\pi(i)$ is 
	\[
		c(\pi(i)) = \min_{k \in \pi(i)} v_{i,k} + \sum_{k \in \pi(i)} c_k .
	\]
	In contrast, the minimum cost following $\pi'(i)$ is $0$ and there is the
	additional cost of the ``values" box. Formally, the cost of $\pi'(i)$ is
	\[
		c(\pi'(i)) =  0 + \sum_{k \in \pi(i)} c_k + c_{(j,v_{i,j})} 
		= \min_{k \in \pi(i)} v_{i,k} + \sum_{k \in \pi(i)} c_k + c_j 
		= c(\pi(i)) + c_j
		\]
	Since $c_j$ appears in the cost of $\pi(i)$, we know that $c(\pi(i)) \geq
	c_j$. Thus, $c(\pi'(i)) = c(\pi(i)) + c_j \leq 2c(\pi(i))$, which implies that 
	$c(\pi') \leq 2c(\pi^*_{\mathcal{I}})$ for our feasible policy $\pi'$. Observing
	that $c(\pi') \geq c(\pi^*_{\mathcal{I}'})$ for any policy, completes the proof.
\end{proof}

\begin{lemma}\label{lem:pbv_to_pb0_policies}
	Given a policy $\pi'$ for the constructed instance $\mathcal{I}'$ of $\pbT$, there exists a feasible
	policy $\pi$ for the instance $\mathcal{I}$ of $\pbv$ with no larger expected cost.
	Furthermore, any branch of $\pi$ can be constructed from $\pi'$ in polynomial time.
\end{lemma}

\begin{proof}[Proof of Lemma~\ref{lem:pbv_to_pb0_policies}]
	We construct a policy $\pi$ for $\mathcal{I}$ using the policy $\pi'$.  Fix
	some branch of $\pi'$. If $\pi'$ opens box $j$ along this branch, we define policy $\pi$ to open the same box along this branch. When $\pi'$ opens a ``final'' box $(j,v)$, we define the policy $\pi$
	to open box $j$ if it has not been opened already. 
 
    Next, we show this policy $\pi$ has no larger expected cost than $\pi'$. There are two cases to
	consider depending on where the ``final'' box $(j,v)$ is opened:
	\begin{enumerate}
		\item ``Final" box $(j, v)$ is at a leaf of $\pi'$:
				since $\pi'$ has finite expected cost and this is the first
				``final" box we encountered, the result must be $0$. Therefore,
				under $\pi$ the values will be $v$ by definition of $\mathcal{I}'$. Observe
				that in this case, $c(\pi) \leq c(\pi')$ since the (at most)
				extra $v$ paid by $\pi$ for the value term, has already been paid by the box cost in $\pi'$
				when box $(j,v)$ was opened.
		\item ``Final" box $(j,v)$ is at an intermediate node of $\pi'$: after $\pi$ opens box $j$, 
				we copy the subtree of $\pi'$ that follows the $0$
				branch into the branch of $\pi$ that follows the $v$ branch. Also, we copy
				the subtree of $\pi'$ that follows the
				$\infty_1$ branch into each branch that has a value different
				from $v$ (the non-$v$ branches).  
				The cost of this new
				subtree is $c_j$ instead of the original $c_j + v$. The $v$ branch may
				accrue an additional cost of $v$ or smaller if $j$ was not the
				smallest values box on this branch, so in total, the
				$v$ branch has cost at most its original cost. 
        
                However, the non-$v$
				branches have a $v$ term removed going down the tree. 
				Specifically, since the feedback of $(j,v)$ down the non-$v$ branch
				was $\infty_1$, some other box with $0$ values had to be
				opened at some point, and this box is still available to be used
				as the final values for this branch later on (since if this
				branch already had a $0$, it would have stopped). Thus, the
				cost of this subtree is at most that originally, and has one
				fewer ``final" box opened. 
	\end{enumerate}
    Putting these cases together implies that $c(\pi) \leq c(\pi')$. 

    Lastly, we argue that any branch of $\pi$ can be computed efficiently. To compute a branch for $\pi$, we follow the corresponding branch of $\pi'$. As we go along this branch, we open box $j$ whenever $\pi'$ opens box $(j,v)$ and remember the feedback. We use the feedback to know which boxes of $\pi'$ to open in the future. Hence, we can compute a branch of $\pi$ from $\pi'$ in polynomial time.
\end{proof}

We are now ready to give the proof of Lemma~\ref{thm:pbv_to_pb0}.
\begin{proof}[Proof of Lemma~\ref{thm:pbv_to_pb0}] 
		Suppose we have an $\alpha$-approximation for $\pbT$. Given an instance
		$\mathcal{I}$ to $\pbv$, we construct the instance $\mathcal{I}'$ for
		$\pbT$ as described and then run the approximation algorithm on
		$\mathcal{I}'$ to get a policy $\pi_\mathcal{I'}$. Next, we prune the tree as
		described in Lemma~\ref{lem:pbv_to_pb0_policies} to get a policy,
		$\pi_\mathcal{I}$ of no worse cost. Our policy will use
		time at most polynomially more than the policy for $\pbT$ since each
		branch of $\pi_\mathcal{I}$ can be computed in polynomial time from $\pi_\mathcal{I'}$. Hence, the
		runtime is polynomial in the size of $\mathcal{I}'$. We also note that
		we added at most $mn$ total ``final'' boxes to construct our new
		instance $\mathcal{I}'$, and so this algorithm will run in polynomial
		time in $m$ and $n$. Thus, by Lemma~\ref{lem:pbv_to_pb0_policies} and
		Lemma~\ref{lem:pbv_to_pb0_opts} we know the cost of the constructed
		policy is
$$c(\pi) \leq c(\pi') \leq \alpha c(\pi^*_{\mathcal{I}'}) \leq 2 \alpha c(\pi^*_{\mathcal{I}})$$
Hence, this algorithm is a $2\alpha$-approximation for $\pbv$.

\end{proof}

\section{Proofs from Section~\ref{sec:pb_to_mssc}}\label{sec:apn_tool}

\msscfToPb*
\begin{proof}[Proof of Claim~\ref{cl:msscf_to_pb}]
Let $\mathcal{I}$ be an instance of $\msscf$. We create an instance $\mathcal{I}'$ of $\pb$ the following way: for every set $s_j$ of $\mathcal{I}$ that gives feedback $f_{ij}$ when element $e_i$ is selected, we create a scenario $s_j$ with the same probability and whose value for box $i$, is either $0$ if $e_i \in s_j$ or  $\infty_{f_{ij}}$ otherwise, where $\infty_{f_{ij}}$ denotes an extremely large value 
which is different for different values of the feedback $f_{ij}$. Observe that any solution to the $\pb$ instance gives a solution to the $\msscf$ at the same cost and vice versa.
\end{proof}

\upbTToUmsscf*
Before formally proving this claim, recall the correspondence of scenarios and boxes of PB-type problems, to elements and sets of \textsc{MSSC}-type problems. The idea for the reduction is to create $T$ copies of sets for each scenario in the initial $\pbT$ instance and one element per box, where if the price a box gives for a scenario $i$ is $<T$ then the corresponding element belongs to all $T$ copies of the set $i$. The final step is to “simulate” the outside option $T$, for which we we create $T$ elements where the $k$’th one belongs only to the $k$’th copy of each set.

\begin{proof}[Proof of Claim~\ref{thm:upbT_to_umsscf}]
  Given an instance $\mathcal{I}$ of $\upbT$ with outside cost box $b_T$, we construct the instance $\mathcal{I}'$ of $\umsscf$ as follows.  
	
	\paragraph{Construction of the instance.} For every scenario $s_i$ in the initial instance, we create $T$ sets denoted by $s_{ik}$ where $k \in [T]$.  Each of these sets has equal probability $p_{ik} = 1/(mT)$. We additionally create one element $e^B$ per box $B$, which belongs to every set $s_{ik}$ for all $k$ iff $v_{Bi} <T$ in the initial instance, otherwise gives feedback $v_{Bi}$.
	In order to simulate box $b_T$ without introducing an element with non-unit cost, we use a sequence of $T$ \emph{outside option} elements $e^T_{k}$ where $e^T_{k}\in s_{ik}$ for all $i\in [m]$ i.e. element $e^T_{ik}$ belongs to ``copy $k$" of every set \footnote{Observe that there are exactly $T$ possible options for $k$ for any set. Choosing all these elements costs $T$ and covers all sets thus simulating $b_T$.}.

	\paragraph{Construction of the policy.}	
	We construct policy $\pi_\mathcal{I}$ by
	ignoring any outside option elements that $\pi_\mathcal{I'}$ selects until
	$\pi_\mathcal{I'}$ has chosen at least $T/2$ such elements, at which point
	$\pi_\mathcal{I}$ takes the outside option box $b_T$. To show
	feasibility we need that for every scenario either $b_T$ is chosen or some box with $v_{ij}\leq T$. If $b_T$ is not chosen, then less than $T/2$ isolating elements were chosen, therefore in instance of $\umsscf$ some sub-sets will have to be covered by another element $e^B$, corresponding to a box. This corresponding box however gives a value $\leq T$ in the initial $\upbT$ instance.

	\paragraph{Approximation ratio.} 

	Let $s_i$ be any scenario in $\mathcal{I}$. We distinguish between the following cases,
	depending on whether there are outside option tests on $s_i$'s branch.
	\begin{enumerate}
	\item \textbf{No outside option tests} on $s_i$'s branch: scenario $s_i$ contributes equally in both policies, since absence of \emph{isolating elements} implies that all copies of scenario $s_i$ will be on the same branch (paying the same cost) in both $\pi_{\mathcal{I}'}$ and $\pi_\mathcal{I}$ 
%
		\item  \textbf{Some outside option tests} on $i$'s branch:  for this
				case, from Lemma~\ref{lem:pbT_to_udt_scens_outside} we have that $c(\pi_{\mathcal{I}}(s_i)) \leq 3 c(\pi_{\mathcal{I}'}(s_i))$.	
	\end{enumerate}

 Putting it all together we get 
	\[
		c(\pi_\mathcal{I}) \leq 3 c(\pi_\mathcal{I'}) 
			\leq 2 \alpha(n+m,m^2) c(\pi^*_\mathcal{I'}) \leq 3\alpha(n+m,m^2) c(\pi^*_\mathcal{I}) ,
	\] 
	where the second inequality follows since we are given an $\alpha$
	approximation and the last inequality since if we are given an optimal
	policy for $\upbT$, the exact same policy is also feasible for
	any $\mathcal{I'}$ instance of $\udt$, which has cost at least $c(\pi^*_\mathcal{I'})$. We also used that $T\leq m$, since otherwise the initial policy would never take the outside option.
\end{proof}

\begin{lemma}\label{lem:pbT_to_udt_scens_outside}
	Let $\mathcal{I}$ be an instance of $\upbT$, and $\mathcal{I'}$ the instance
	of $\umsscf$ constructed by the reduction of
	Claim~\ref{thm:upbT_to_umsscf}. For a scenario $s_i$, if there is at least one outside option test run in
	$\pi_\mathcal{I}$, then $c(\pi_\mathcal{I}(s_i)) \leq 3c(\pi_\mathcal{I'}(s_i))$.
\end{lemma}

	\begin{proof}
		For the branch of scenario $s_i$, denote by $M$ the box elements chosen before there were $T/2$ \emph{outside
			option} elements, and by $N$ the number of \emph{outside option} elements in $\pi_{\mathcal{I}'}$. Note that the smallest cost is achieved if  all the outside option elements are chosen
		first\footnote{Since the outside option tests cause some copies to be
		isolated and so can reduce their cost.}.
			The copies of scenario $s_i$ can be split into two groups;
		those that were isolated \textbf{before} $T/2$ \emph{outside option} elements were chosen,
		and those that were isolated \textbf{after}. We distinguish between the
		following cases, based on the value of $N$.
	
	\begin{enumerate}
		\item $N \geq T/2$: in this case each of the copies of $s_i$ that are isolated after pays at least $M
				+ T/2$ for the initial box elements and the initial sequence of
				\emph{outside option} elements. For the copies isolated before, we lower bound the
				cost by choosing all \emph{outside option} elements first.

    The cost of all the copies in $\pi_{\mathcal{I}'}$ then is at least
		\begin{align*}
		    \sum_{j = 1}^{K_i} \sum_{k = 1}^{T/2} \frac{c p_{\ell}}{T}k + \sum_{j = 1}^{K_i} \sum_{k = T/2+1}^T \frac{c p_{\ell}}{T} (T/2+M) &= cp_i \frac{ \frac{T}{2}(\frac{T}{2}+1)}{2T}  + c p_i\frac{\frac{T}{2}(T/2+M)}{T}\\
		    &\geq cp_i(3T/8 + M/2) \\
	   	 &\geq \frac{3}{8}p_i(T + M)
		\end{align*}
		Since $N \geq T/2$, policy $\pi_{\mathcal{I}}$ will take the
		outside option box for $s_i$, 
  immediately after choosing the $M$ initial boxes corresponding to the box elements.
		So, the total contribution $s_i$ has on the expected cost of $\pi_{\mathcal{I}}$ is
		at most $p_i (M+T)$ in this case. Hence, we have that $s_i$'s
		contribution in $\pi_\mathcal{I}$ is at most $\frac{8}{3} \leq 3$ times $s_i$'s
		contribution in $\pi_{\mathcal{I}'}$.
	
		\item $N < T/2$: policy $\pi_{\mathcal{I}}$ will only select
				the $M$ boxes (corresponding to \emph{box} elements) and this was sufficient for finding a value
				less than $T$. The total contribution of $s_i$ on $c(\pi_\mathcal{I})$ is exactly $p_i M$. On the other
				hand, since $N < T/2$ we know that at least half of the
				copies will pay $M$ for all of the box elements. The cost of all the copies is at least
			\[ 
   \sum_{j = 1}^{K_i} \sum_{k = N}^T \frac{cp_{\ell}}{T}M = cp_i \frac{T-N}{T}M \geq cp_i M/2 ,
   \]
				therefore, the contribution $s_i$ has on $c( \pi_{\mathcal{I}'}) $ is at least $cp_i M/2$. Hence, we have $c(\pi_\mathcal{I}) \leq 3
				c(\pi_{\mathcal{I}'})$ 
		\end{enumerate}
	\end{proof}

\subsection{Proofs from subsection~\ref{subsec:main_tool}}\label{subsec:apn_main_tool}

\thresholdBound*
\begin{proof}
    Consider a policy $\pi_\iq$ which runs $\pi^*_\mathcal{I}$ on the instance $\iq$; and for scenarios with cost $c_s\ge t_{q/(\beta\alpha)}$ aborts after spending this cost and chooses the outside option $T$. The cost of this policy is:
 
	\begin{equation}\label{eq:optz}
			c(\pi^*_\iq) \leq  c(\pi_\iq)
			= \frac{T+ t_{q/(\beta\alpha)}}{\beta \alpha}  + \sum_{ 
			\substack{ c_s \in [t_q, t_{q/(10\alpha)}]\\ s\in \scenarios }} c_s \frac{p_s}{q}  ,
	\end{equation}
	\noindent
    By our assumption on $T$, this cost is at most $2T/\beta\alpha$.
 On the other hand
	since $\mathcal{A}_T$ is an $\alpha$-approximation to the optimal we have
	that the cost of the algorithm's solution is at most
	\[
		\alpha c(\pi^*_\iq) \leq \frac{2T}{\beta}
	\]
Since the expected cost of $\mathcal{A}_T$ is at most
	$2T/\beta$, then using Markov's inequality, we get that $\Pr{}{c_s \geq T} \leq (2T/\beta)/T
	= 2/\beta$. Therefore, $\mathcal{A}_T$ covers at least $1-2/\beta$ mass every time. 
\end{proof}

\pbvOpt*
\begin{proof}
	In every interval of the form $\mathcal{I}_i =[t_{q^i},
	t_{q^i/(\beta\alpha)} ]$ the optimal policy for $\pbv$ covers at least
	$1/(\beta\alpha)$ of the probability mass that remains. Since the values
	belong in the interval $\mathcal{I}_i$ in phase $i$, it follows that the
	minimum possible value that the optimal policy might pay is $t_{q^i}$,
	i.e. the lower end of the interval. Summing up for all intervals, we get
	the lemma.
\end{proof}

\subsection{Proofs from subsection~\ref{subsec:improved_tool}}\label{sec:apn_new_tool}

\begin{algorithm}[H]
	\KwIn{Set of scenarios $\scenarios$}
	Scale all probabilities by $c$ such that $c \sum_{s\in\scenarios} p_s = 1$\\
	Let $p_{\text{min}} = \min_{s\in\scenarios} p_s$\\
	$\scenarios' = $ for each $s\in \scenarios$ create $p_s/p_{\text{min}}$ copies\\
	Each copy has probability $1/|\scenarios'|$\\
	\Return $\scenarios'$
\caption{\textbf{Expand}: rescales and returns an instance of $\upb$.}\label{algo:expand}
\end{algorithm}
	
\pbvToUpbLog*
\begin{proof}
The proof in this case follows the steps of the proof of Theorem~\ref{thm:pbv_to_pb0_log}, and we are only highlighting the changes. The process of the reduction is the same as Algorithm~\ref{algo:pbv_pb0} with the only difference that we add two extra steps; (1) we initially remove all low probability scenarios (line 3 - remove at most $c$ fraction) and (2) we add them back after running $\upbT$ (line 8). The reduction process is formally shown in Algorithm~\ref{algo:pbv_upb0}.

	\paragraph{Calculating the thresholds.}
	For every phase $i$ we choose a threshold $T_i$ such that $T_i = \min \{T :
	\Pr{}{c_s>T} \leq \delta\}$ i.e. at most $\delta$ of the probability mass of
	the scenarios are not covered, again using binary search as in Algorithm~\ref{algo:pbv_pb0}.  We denote by
	$\intv_i = [t_{(1-c)(\delta+c)^i}, t_{(1-c)(\delta+c)^i/(\beta \alpha)} ] $ the relevant interval of costs at every run of the algorithm, then by
	Lemma~\ref{lem:thresh_bound}, we know that for remaining total probability
	mass $(1-c)(\delta+c)^i$, any threshold which satisfies
	\begin{equation*}
		T_i \geq t_{(1-c)(\delta+c)^{i-1}/\beta \alpha} 
			+ \beta  \alpha \sum_{\substack{s\in \scenarios\\ c_s \in \intv_i }} c_s \frac{p_s}{(1-c)(\delta+c)^i} 
	\end{equation*}
	also satisfies the desired covering property, i.e. at least $(1-2/\beta)(1-c)(\delta+c)$ mass of the
	current scenarios is covered. Therefore the threshold
	$T_i$ found by our binary search satisfies 
 \begin{equation}\label{eq:threshold_uniform}
		T_i = t_{(1-c)(\delta+c)^{i-1}/\beta \alpha} 
			+ \beta \alpha \sum_{\substack{s\in \scenarios\\ c_s \in \intv_i }} c_s \frac{p_s}{(1-c)(\delta+c)^i} .
	\end{equation}

Following the proof of Theorem~\ref{thm:pbv_to_pb0_log}, \textbf{Constructing the final policy} and \textbf{Accounting for the values} remain exactly the same as neither of them uses the fact that the scenarios are uniform.

	\paragraph{Bounding the final cost.}
	Using the guarantee that at the end of every phase we cover $(\delta+c)$ of the
	scenarios, observe that the algorithm for $\pbT$ is run in an interval of the
	form $\intv_i = [t_{(1-c)(\delta+c)^i}, t_{(1-c)(\delta+c)^i/(\beta \alpha)} ]$. Note also
	that these intervals are overlapping. Bounding the cost of the final policy
	$\pi_\mathcal{I}$ for all intervals we get
	\begin{align*}
			\pi_\mathcal{I} &\leq  \sum_{i=0}^\infty (1-c)(\delta+c)^i T_i & \\
				 &= \sum_{i=0}^\infty \lp( (1-c)(\delta+c)^i t_{(1-c)(\delta+c)^{i-1}/\beta \alpha} 
				 + \beta  \alpha \sum_{\substack{s\in \scenarios\\ c_s \in \intv_i }} c_s p_s \rp)
			 		& \text{From equation \eqref{eq:threshold_uniform}} \\
			& \leq 2\cdot \beta\alpha \pi^*_\mathcal{I} + 
				\beta  \alpha \sum_{i=0}^\infty  \sum_{\substack{s\in \scenarios\\ c_s \in \intv_i }} c_s p_s 
					&\text{Using Lemma~\ref{lem:pbv_opt}}\\
			&\leq 2\beta \alpha \log \alpha \cdot \pi^*_\mathcal{I}, & 
	\end{align*}
	 where the inequalities follow similarly to the proof of Theorem~\ref{thm:pbv_to_pb0_log}. Choosing $c=\delta = 0.1$ and $\beta = 20$ we get the theorem.
\end{proof}

\section{Proofs from Section~\ref{sec:mssc_to_dt}}\label{sec:apn_mssc_and_odt}
\msscfToOdt*
\begin{proof}[Proof of Claim~\ref{thm:msscf_to_odt}]
    Let $\mathcal{I}$ be an instance of $\msscf$. We create an instance $\mathcal{I}'$ of $\odt$ the following way: for every set $s_j$ we create a scenario $s_j$ with the same probability and for every element $e_i$ we create a test $T_{e_i}$ with the same cost, that gives full feedback whenever an element belongs to a set, otherwise returns only the element's feedback $f_{ij}$. Formally, the $i$-test under scenario $s_j$ returns
    \[T_{e_i}(s_j) = \begin{cases}
        \text{``The feedback is $f_{ij}$''} & \text{If } e_i\not\in s_j \\
        \text{``The scenario is $j$''} & \text{else },
    \end{cases}\]
    therefore the test isolates scenario $j$ when $e_i\in s_j$.
    
\paragraph{Constructing the policy.} Given a policy $\pi'$ for the instance $\mathcal{I}'$ of $\odt$, we can construct a policy $\pi$ for $\mathcal{I}$ by selecting the element that corresponds to the test $\pi'$ chose. When $\pi'$ finishes, all scenarios are identified and for any scenario $s_j$ either (1) there is a test in $\pi'$ that corresponds to an element in $s_j$ (in the instance $\mathcal{I}$) or (2) there is no such test, but we can pay an extra $\min_{i\in s_j} c_i$ to select the lowest cost element in this set\footnote{Since the scenario is identified, we know exactly which element this is.}.

Observe also that in this instance of $\odt$ if we were given the optimal solution, it directly translates to 
\emph{a} solution for $\msscf$ with the same cost, therefore 
\begin{equation}\label{eq:opt_msscf_to_odt}
    c(\pi^*_\mathcal{I}) \leq c(\pi'_{\mathcal{I}'}) = c(\pi^*_{\mathcal{I}'})
\end{equation}

\paragraph{Bounding the cost of the policy.} As we described above the total cost of the policy is 
\begin{align*}
    c(\pi) & \leq c(\pi_{\mathcal{I}'}) + \E{s\in\scenarios}{\min_{i\in s} c_i} \\
    & \leq c(\pi_{\mathcal{I}'}) + c(\pi^*_\mathcal{I}) \\
    &\leq a(n,m)c(\pi^*_{\mathcal{I}'}) + c(\pi^*_\mathcal{I}) \\
    & = (1+a(n,m)) c(\pi^*_\mathcal{I}),
\end{align*}
where in the last inequality we used equation~\eqref{eq:opt_msscf_to_odt}.

Note that for this reduction we did not change the probabilities of the scenarios, therefore if we had started with uniform probabilities and had an oracle to $\udt$, we would still get an $a(n,m)+1$ algorithm for $\umsscf$.
\end{proof}


In the reduction proof of Theorem~\ref{thm:udt_to_pbT}, we are using the following two lemmas, that show
that the policy constructed for $\udt$ via the reduction is feasible and has bounded cost. 

\begin{lemma}\label{lem:udt_to_pbT_feasible}
		Given an instance $\mathcal{I'}$ of $\udt$ and the corresponding
		instance $\mathcal{I}$ of $\umsscf$ in the reduction of
		Theorem~\ref{thm:udt_to_pbT}, the policy $\pi_\mathcal{I'}$ constructed
		for $\udt$ is feasible.
	\end{lemma}
	\begin{proof}[Proof of Lemma~\ref{lem:udt_to_pbT_feasible}]
		It suffices to show that every scenario is isolated. Fix a scenario $s_i$.
		Observe that $s_i$'s branch has chosen the isolating element $E^i$ in the $\umsscf$ solution, since that is the the only element that belongs to set $s_i$.
  Let $S$ be the set of scenarios 
			just before $E^i$ is chosen and note that by definition $s_i\in S$. 

			If $|S|=1$, then since $\pi_{\mathcal{I}'}$ runs tests giving the
			same branching behavior by definition of $\pi_{\mathcal{I}'}$, and
			$s_i$ is the only scenario left, we have that the branch of
			$\pi_{\mathcal{I}'}$ isolates scenario $s_i$. 

		If $|S|>1$ then all scenarios/sets in $S \setminus \{s_i\}$ are not covered by choosing element $E^i$, therefore they are covered at strictly deeper leaves in the tree. By induction on
	the depth of the tree, we can assume that for each scenario $s_j \in \lp( S
	\setminus \{s_i\} \rp)$ is isolated in $\pi_{\mathcal{I}'}$.
	We distinguish the following cases based on when we encounter
			$E^i$ among the isolating elements on $s_i$'s branch.
	
	\begin{enumerate}
			\item \textbf{$E^i$ was the first isolating element chosen on the
	branch}: then policy $\pi_{\mathcal{I}'}$ ignores element $E^i$.
					Since every leaf holds a unique scenario in $S \setminus
					\{s_i\}$, if we ignore $s_i$ it follows some path of
					tests and either be isolated or end up in a node that
					originally would have had only one scenario, as shown in Figure~\ref{fig:apnx_udt_to_pbT_1}. Since there
					are only two scenarios at that node, policy
					$\pi_{\mathcal{I}'}$ runs the cheapest test
					distinguishing $s_i$ from that scenario. 
		\begin{figure}[H]
		\centering
		\begin{tikzpicture}
\tikzset{My Style/.style={ellipse, draw=black, minimum size=25pt}}

\pgfmathsetmacro{\spacing}{2}
\pgfmathsetmacro{\offsetX}{9}

\node [My Style,label={[yshift=-0.6cm,xshift=1cm]\textcolor{red}{$E^i$}}] (root) at (0,0) {$\mathcal{S}$};

\node[My Style] (left) at (-1,-\spacing) {$s_i$};
\node[My Style] (right) at (1,-\spacing) {$\mathcal{S}\setminus s_i$};

\draw (right) node[]{}
  -- (-1, -3*\spacing) node[]{}
  -- (3,-3*\spacing) node[]{}
  -- (right);

\node[inner sep=0pt] (invisible) at (1, -3*\spacing) {};
\node[My Style] (leaf) at (1, -3.7*\spacing) {$s_{\text{leaf}}$};

\node[My Style] (randj) at (-1.5, -3.7*\spacing) {$s_j$};
\draw[-] (-0.7,-3*\spacing) -- (randj);

\node[] (dots) at (-0.5, -3.7*\spacing) {$\ldots$};
\node[] (dots2) at (2.5, -3.7*\spacing) {$\ldots$};

\node[My Style] (randk) at (3.5, -3.7*\spacing) {$s_k$};
\draw[-] (2.6,-3*\spacing) -- (randk);

\path[draw = red,decorate, decoration=snake] (right)--(invisible);

\draw[-] (left) -- (root);
\draw[-] (invisible) -- (leaf);
\draw[-] (right) -- (root);

		\draw[thick, ->] (3, -1.5) -- (7, -1.5) node[label={[xshift=-2cm]{Ignoring Box $E^i$}}] {};

\node[My Style] (right2) at (1+\offsetX,-\spacing) {$\mathcal{S}$};

\draw (right2) node[]{}
  -- (-1+\offsetX, -3*\spacing) node[]{}
  -- (3+\offsetX,-3*\spacing) node[]{}
  -- (right2);

\node[inner sep=0pt] (invisible2) at (1+\offsetX, -3*\spacing) {};
\node[My Style] (leaf2) at (1+\offsetX, -3.7*\spacing) {$s_{\text{leaf}},$ \textcolor{red}{$s_i$}};

\path[draw = red,decorate, decoration=snake] (right2)--(invisible2);

\node[My Style] (randj2) at (-1.5+\offsetX, -3.7*\spacing) {$s_j$};
\draw[-] (-0.7+\offsetX,-3*\spacing) -- (randj2);

\node[] (dotsb) at (-0.5+\offsetX, -3.7*\spacing) {$\ldots$};
\node[] (dots2b) at (2.5+\offsetX, -3.7*\spacing) {$\ldots$};

\node[My Style] (randk2) at (3.5+\offsetX, -3.7*\spacing) {$s_k$};
\draw[-] (2.6+\offsetX,-3*\spacing) -- (randk2);

\draw[-] (invisible2) -- (leaf2);

\end{tikzpicture}
			\caption{Case 1: $\mathcal{S}$ is the set of scenarios remaining 
					when $E^i$ is chosen, $s_{\text{leaf}}$ is the scenario that $s_i$ ends up with.}
		\label{fig:apnx_udt_to_pbT_1}
	\end{figure}
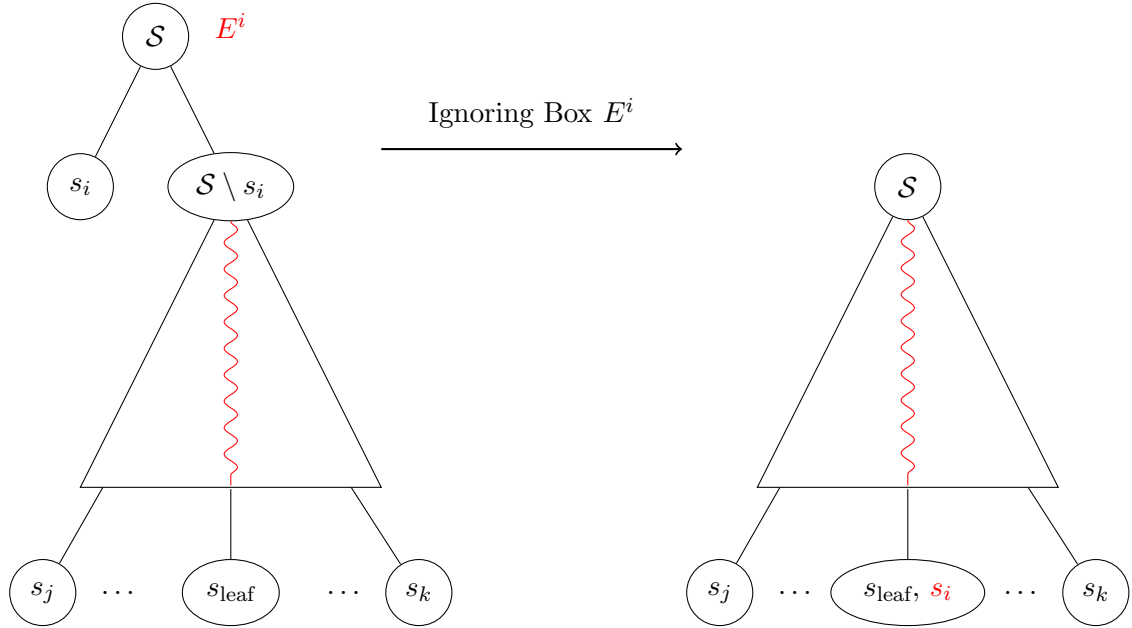    

		\item \textbf{A different element $E^j$ was chosen before $E^i$}:
				by our construction, instead of ignoring $E^i$ we
				now run the cheapest test that distinguishes $s_i$ from $s_j$, causing 
				$i$ and $j$ to go down separate branches, as shown in figure~\ref{fig:apnx_udt_to_pbT_2}. We apply the induction
				hypothesis again to the scenarios in these sub-branches, therefore, both $s_i$ and
				$s_j$ are either isolated or end up in a node with a single
				scenario and then get distinguished by the last case of
					$\pi_{\mathcal{I}'}$'s construction.

			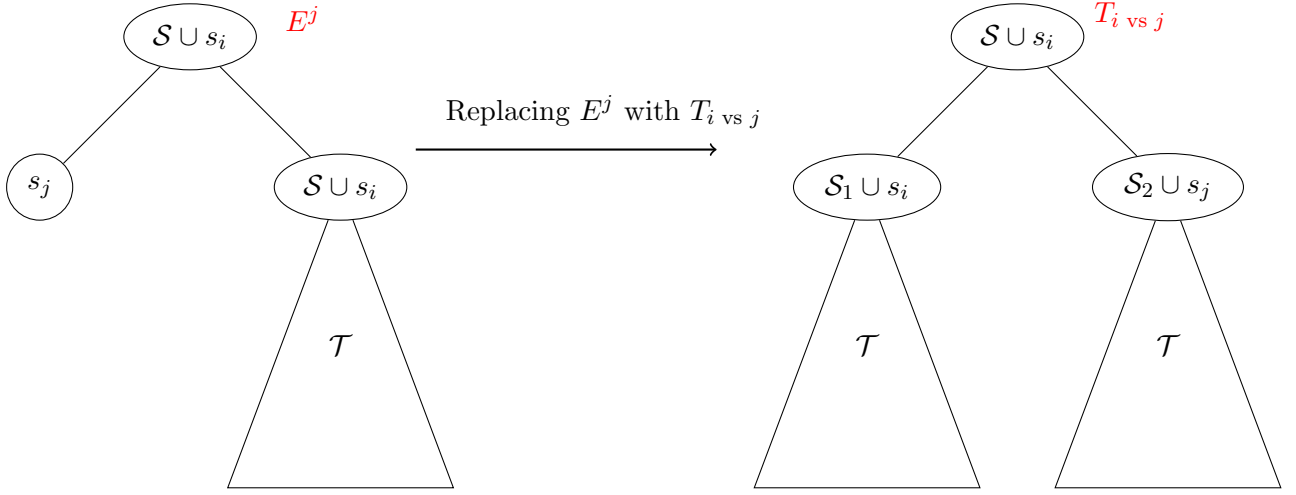
\begin{figure}[h]
				\centering
				\begin{tikzpicture}
\tikzset{My Style/.style={ellipse, draw=black, minimum size=25pt}}

\pgfmathsetmacro{\spacing}{2}
\pgfmathsetmacro{\offsetX}{11}

\node [My Style,label={[yshift=-0.5cm,xshift=1.5cm]\textcolor{red}{$E^j$}}] (root) at (0,0) {$\mathcal{S}\cup s_i$};

\node[My Style] (left) at (-2,-\spacing) {$s_j$};
\node[My Style] (right) at (2,-\spacing) {$\mathcal{S}\cup s_i$};

		\draw (right) node[label={[xshift=0cm, yshift=-2.5cm]$\mathcal{T}$}]{}
  -- (0.5, -3*\spacing) node[]{}
  -- (3.5,-3*\spacing) node[]{}
  -- (right);

 \draw[-] (root) -- (right);
 \draw[-] (root) -- (left);

\draw[thick, ->] (3, -1.5) -- (7, -1.5) node[label={[xshift=-1.5cm, yshift=0cm]Replacing $E^j$ with $T_{i \text{ vs }j}$}] {};

\node [My Style,label={[yshift=-0.5cm,xshift=1.5cm]\textcolor{red}{$T_{i \text{ vs }j}$}}] (root) at (0+\offsetX,0) {$\mathcal{S}\cup s_i$};

\node[My Style] (left2) at (-2+\offsetX,-\spacing) {$\mathcal{S}_1 \cup s_i$};
\node[My Style] (right2) at (2+\offsetX,-\spacing) {$\mathcal{S}_2\cup s_j$};

\draw (right2) node[label={[xshift=0cm, yshift=-2.5cm]$\mathcal{T}$}]{}
  -- (0.5+\offsetX, -3*\spacing) node[]{}
  -- (3.5+\offsetX,-3*\spacing) node[]{}
  -- (right2);

\draw (left2) node[label={[xshift=0cm, yshift=-2.5cm]$\mathcal{T}$}]{}
  -- (-3.5+\offsetX, -3*\spacing) node[]{}
  -- (-0.5+\offsetX,-3*\spacing) node[]{}
  -- (left2);

 \draw[-] (root) -- (right2);
 \draw[-] (root) -- (left2);
\end{tikzpicture}
					\caption{Case 2: run test $T_{i \text{ vs } j}$ to distinguish $s_i$ and $s_j$. Sets $\mathcal{S}_1$ and $
						\mathcal{S}_2$ partition $\mathcal{S}$ } 
					\label{fig:apnx_udt_to_pbT_2}
			\end{figure}
	\end{enumerate}

	Hence, $\pi_{\mathcal{I}'}$ is isolating for any scenario $s_i$.
	Also, notice that any two scenarios that have isolating boxes on the same
	branch will end up in distinct subtrees of the lower node.
	\end{proof}

	\begin{lemma}\label{lem:udt_to_pbT_approx}
		Given an instance $\mathcal{I}$ of $\umsscf$ and an instance
		$\mathcal{I'}$ of $\udt$, in the reduction of
		Theorem~\ref{thm:udt_to_pbT} it holds that 
		\[c(\pi_\mathcal{I'}) \leq 2 c(\pi_\mathcal{I}).
			\]
	\end{lemma}
	\begin{proof}[Proof of Lemma~\ref{lem:udt_to_pbT_approx}]
		Let $s_i$ be any scenario in $\mathcal{S}$. We use induction on the
			number of isolating boxes along $s_i$'s branch in $\mathcal{I}'$.
			Initially observe that $E^i$ will always exist
			in $s_i$'s branch, in any feasible solution to $\mathcal{I}$. We
			use $c(E^j)$ and $c(T_k)$ to denote the costs of box $E^j$ and test
			$T_k$, for any $k\in[n]$ and $j\in [n+m]$.

		\begin{enumerate}
			\item \textbf{Only $E^i$ is on the branch}: since $E^i$ will be
				ignored, we end up with $s_i$ and some other not yet isolated scenario, let $s_{\text{leaf}}$ be that scenario. To
				isolate $s_i$ and $s_{\text{leaf}}$ we run the cheapest test that
				distinguishes between these. From the definition of the cost of
						$E^i$ we know that $c(T_{s_i \text{ vs }
						s_{\text{leaf}}}) \leq c(E^i)$. Additionally, since
						$c(s_i) \leq c(s_{\text{leaf}})$ and both $s_{\text{leaf}}$ and $s_i$ have probability $1/m$, overall we have
						$c(\pi_\mathcal{I}) \leq 2c(\pi_\mathcal{I'})$. This is also shown in Figure~\ref{fig:apnx_udt_to_pbT_1}.

			\item \textbf{More than one isolating elements are on the
					branch}: similarly, observe
						that for any extra isolating element $E^j$ we encounter, we
						substitute it with a test that distinguishes between
						$s_i$ and $s_j$ and costs at most $c(E^j)$. Given that $c(s_i) \leq c(s_{\text{leaf}})$ and scenarios are uniform, we again have 
						$c(\pi_\mathcal{I}) \leq 2c(\pi_\mathcal{I'})$.
		\end{enumerate}
 
	\end{proof}

\udtTopbT*
	\begin{proof}[Proof of Theorem~\ref{thm:udt_to_pbT}]
 We begin by giving the construction of the policy in the reduction, and showing the final approximation ratio. 
	\paragraph{Constructing the policy.} Given a policy
	$\pi_{\mathcal{I}}$ for the instance of $\umsscf$, we construct a policy
	$\pi_{\mathcal{I}'}$. For any test element
	$B_j$ that $\pi_{\mathcal{I}}$ selects, $\pi_{\mathcal{I}'}$ runs the
	equivalent test $T_j$. For the \emph{isolating elements} $E^i$ we distinguish the following cases.
	
	\begin{enumerate}
		\item If $\pi_{\mathcal{I}}$ selects an \emph{isolating element} $E^i$ for the first time on the current branch, then $\pi_{\mathcal{I}'}$ ignores
				this element but remembers the set/scenario $s_i$, which $E^i$ belonged to.
		\item If $\pi_{\mathcal{I}}$ selects another \emph{isolating element} $E^j$
				after some $E^i$ on the branch, then $\pi_{\mathcal{I}'}$ runs the minimum
				cost test that distinguishes scenario $s_j$ from $s_k$ where $E^k$ was the most recent \emph{isolating element} chosen on this branch prior
				to $E^j$. 
		\item If we are at the end of $\pi_{\mathcal{I}}$, there can be at most
				$2$ scenarios remaining on the branch, so $\pi_{\mathcal{I}'}$
				 runs the minimum cost test that distinguishes these two
				scenarios.
	\end{enumerate}
	By Lemma~\ref{lem:udt_to_pbT_feasible}, we have that the above policy is feasible for $\udt$.

	\paragraph{Approximation ratio.} From Lemma~\ref{lem:udt_to_pbT_approx} we
	have that $c(\pi_\mathcal{I'}) \leq 2 c(\pi_\mathcal{I})$. For
	the optimal policy, we have that $c(\pi^*_{\mathcal{I}}) \leq
	3c(\pi^*_{\mathcal{I}'})$. This holds since if we have an optimal
	solution to $\udt$, we can add an \emph{isolating element} at every leaf to make it
	feasible for $\umsscf$, by only increasing the cost by a factor of
	$3$\footnote{This is because for every two scenarios, the $\udt$ solution
			must distinguish between them, but one of these scenarios is the
	$\max$ scenario from the definition of $T_j$, for which we pay less than
	$T_j$}, which means that $c(\pi^*_{\mathcal{I}})$ will be less than this
	transformed $\umsscf$ solution. Overall, if $\pi_{\mathcal{I}}$ is computed
	from an $\alpha(n,m)$-approximation for $\umsscf$, we have \[
    c(\pi_{\mathcal{I}'})
	\leq 2 c(\pi_{\mathcal{I}}) \leq 2 \alpha(n+m,m) c(\pi_{\mathcal{I}}^*) \leq  6
	\alpha(n+m,m) c(\pi_{\mathcal{I}'}^*) \] 
 

\end{proof}

\section{Proofs from Section~\ref{sec:mixt}}\label{sec:apn_mixt}
\dpTests*
	\begin{proof}
		Let $s_1, s_2\in \scenarios$ be any two scenarios in the instance of
		$\pbv$ and let $v_i$ be the value returned by opening the $i$'th
		informative box, which has distributions $\dsa$ and $\dsb$ for
		scenarios $s_1$ and $s_2$ respectively. Then by the definition of
		informative boxes for every	such box opened, there is a set of values
		$v$ for which $\Pr{\dsa}{v} \geq \Pr{\dsb}{v}$  and a set for which the
		reverse holds. Denote these sets by $M_{i}^{s_1}$ and $M_{i}^{s_2}$
		respectively. We also define the indicator variables $X_i^{s_1} =
		\ind{v_i \in M_i^{s_1}}$. 
		Define $\overline{X} = \sum_{i\in [k]}X_i^{s_1}/k$, and observe that
		$\E{}{\overline{X}| s_1} = \sum_{i\in [k]} \Pr{}{M_i^{s_1}}/k$.
		Since for every box we have an $\e$ gap in TV distance between the scenarios $s_1, s_2$ we have that 
			\[
			\lp| \E{}{\overline{X}|s_1} - \E{}{\overline{X}|s_2} \rp| \geq \e,
			\]
			 therefore if $\lp| \overline{X} - \E{}{\overline{X}|s_1}\rp| \leq \e/2$ we conclude that scenario $s_2$ is
		eliminated, otherwise we eliminate scenario $s_1$.  
		The probability of error is $\Pr{\dsa}{\overline{X} - \E{}{\overline{X}|s_1} >
		\e/2} \leq e^{-2k(\e/2)^2}$, where we used Hoeffding's inequality since $X_i \in
		\{0,1\}$. Since we want the probability of error to be less than
		$\delta$, we need to open $O\lp( \frac{\log 1/\delta}{\e^2} \rp)$ informative
		boxes.
	\end{proof}

		\begin{proof}[Proof of Theorem~\ref{thm:dp}]
		We describe how to bound the final cost, and calculate the runtime of
		the DP. Denote by $L  = m^2/\e^2 \log 1/\delta$ where we show that in
		order to get $(1+\beta)$-approximation we set $\delta = \frac{\beta
				c_{\min}}{m^2 T}$.

		\paragraph{Cost of the final solution.}	Observe that the only case where the
		DP limits the search space is when $|S|=1$. If the scenario is identified
		correctly, the DP finds the optimal solution by
		running the greedy order; every time choosing the box with the highest
		probability of a value below $T$\footnote{When there is only one
		scenario, this is exactly Weitzman's algorithm.}. 

		In order to eliminate all scenarios but one, we should eliminate all
		but one of the $m^2$ pairs in the list $E$. From
		Lemma~\ref{lem:dp_tests}, and a union bound on all $m^2$ pairs, the
		probability of the last scenario being the wrong one is at most
		$m^2\delta$. By setting $\delta = \beta c_{\min}/(m^2T)$, we get that
		the probability of error is at most $\beta c_{\min}/T$, in which case
		we pay at most $T$, therefore getting an extra $\beta c_{\min} \leq \beta	c(\pi^*)$
		factor.
		
		\paragraph{Runtime.} The DP maintains a list $M$ of sets  of
		informative boxes opened, and numbers of non informative ones. Recall
		that $M$ has the following form $M = S_1|x_1 |S_2| x_2| \ldots |S_k
		|x_k$, where $k\leq L$ from Lemma~\ref{lem:dp_tests} and the fact that there are $m^2$ pairs in $E$. There are in
		total $n$ boxes, and $L$ ``positions" for them, therefore the size of
		the state space is ${n \choose L} = O(n^L)$. There is also an extra $n$
		factor for searching in the list of informative boxes at every step of
		the recursion. Observe that the numbers of non-informative boxes also
		add a factor of at most $n$ in the state space.  The list $E$ adds
			another factor at most $n^{m^2}$, and the list $S$ a factor of $ 2^m$
			making the total runtime to be $ n^{\tilde{O}(m^2/\e^2)}$.
	\end{proof}

\section{Boxes with Non-Unit Costs: Revisiting our Results}\label{apn:general_costs}
In the original \pbText{} problem, denoted by $\pbvc$, each box $i$ has a different known cost $c_i>0$. Similarly we denote the non-unit cost version of both decision tree-like problems and Min Sum Set Cover-like problems by adding a superscript ${}^c$ to the problem name. Specifically, we now define $\odtc$, $\udtc$, $\msscfc$ and $\umsscfc$, where the tests (elements) have non-unit cost for the decision tree (min sum set cover) problems. We revisit our results and describe how our reductions change to incorporate non-unit cost boxes (summary in Figure~\ref{fig:summary_costs}).

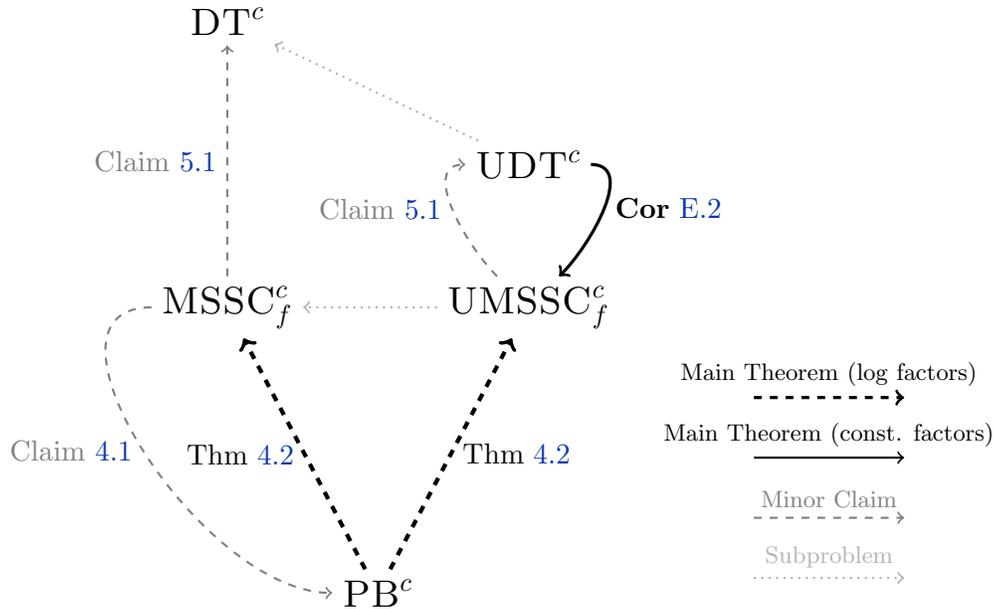
\begin{figure}[H]
    \centering
    \newcommand{\bl}[1]{\textcolor{black}{#1}}
\pgfmathsetmacro{\dist}{3}
\pgfmathsetmacro{\distV}{1.9}
\pgfmathsetmacro{\offset}{1}
\pgfmathsetmacro{\ypomnimaX}{6}
\pgfmathsetmacro{\ypomnimaY}{1.8}
\begin{tikzpicture}

	\node (pbv) at (0,0){{\Large $\pbvc$}};

	\node (umsscf) at (2*\offset,2*\distV){{\Large $\umsscfc$}};
 	\node (msscf) at (-2*\offset,2*\distV){{\Large $\msscfc$}};

	\node (udt) at (2*\offset, 3*\distV){{\Large $\udtc$}};
 	\node (odt) at (-2*\offset, 4*\distV){{\Large $\odtc$}};
	\draw[->,dotted, gray, thick, opacity=0.6] (umsscf) -- (msscf);

\draw[->,dotted, gray, thick, opacity=0.6] (udt) --  (odt);

    \draw[->, dashed, thick, gray] (msscf) to [out=180, in=180]  node [above, left] {Claim~\ref{cl:msscf_to_pb}}  (pbv);
 
    \draw[->,dashed, thick, gray] (msscf) -- node [above, left] {Claim~\ref{thm:msscf_to_odt}} (odt); 
    
    \draw[->, dashed, thick, gray] (umsscf) to [out=135,in=180]  node [above, left] {Claim~\ref{thm:msscf_to_odt}}  (udt);

    \draw[->,ultra thick, dashed] (pbv) -- node [above, right, text width=2cm] {Thm~\ref{thm:pb_to_mssc}}(umsscf);
        
    \draw[->,very thick] (udt) to [in=45, out=0] node [above, right] {\bl{\textbf{Cor}~}\ref{cor:udt_to_pbT_costs}}  (umsscf); 
    
     \draw[->,ultra thick, dashed] (pbv) -- node [label={[xshift=-0.85cm, yshift=-0.45cm]Thm~\ref{thm:pb_to_mssc}}] {} (msscf);

    \draw[->,very thick, dashed] (\ypomnimaX-1, \ypomnimaY+0.8) -- node[above] {{\footnotesize Main Theorem ($\log$ factors)}} (\ypomnimaX+1, \ypomnimaY+0.8);
    
        \draw[->,thick] (\ypomnimaX-1, \ypomnimaY) -- node[above] {{\footnotesize Main Theorem (const. factors)}} (\ypomnimaX+1, \ypomnimaY);
        
    \draw[->,dashed,thick, gray] (\ypomnimaX-1, \ypomnimaY-0.8) -- node[above] {{\footnotesize Minor Claim}} (\ypomnimaX+1, \ypomnimaY-0.8);
    
	  \draw[->,dotted, gray, thick, opacity=0.6] (\ypomnimaX-1, \ypomnimaY-1.6) -- node[above] {{\footnotesize Subproblem}} (\ypomnimaX+1, \ypomnimaY-1.6);

\end{tikzpicture}
    \caption{Summary of all the reductions with non-unit costs. The only result that needs a changed proof is Corollary~\ref{cor:udt_to_pbT_costs} highlighted in bold (previously Theorem~\ref{thm:udt_to_pbT}).}
    \label{fig:summary_costs}
\end{figure}

Note also, that even though the known results for \dtText{} (e.g. \cite{GuilBilm2009,GuptNagaRavi2017}) handle non-unit test costs, the currently known works for \udtText{} do not. If however there is an algorithm for \udtText{} with non-unit costs, our reductions will handle this obtaining the same approximation guarantees.

\subsection{Connecting Pandora's Box and $\msscf$}

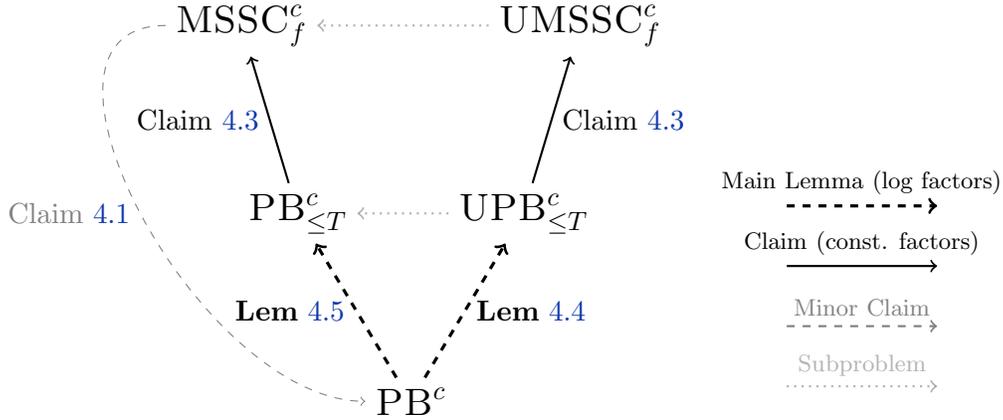
\begin{figure}[H]
    \centering
    \newcommand{\bl}[1]{\textcolor{black}{#1}}
\pgfmathsetmacro{\dist}{3}
\pgfmathsetmacro{\distV}{2.5}
\pgfmathsetmacro{\offset}{1.5}
\pgfmathsetmacro{\ypomnimaX}{6}
\pgfmathsetmacro{\ypomnimaY}{1.8}
\begin{tikzpicture}

	\node (pbv) at (0,0){{\Large $\pbvc$}};

	\node (upbT) at (\offset,\distV){{\Large $\upbTc$}};
 	\node (pbT) at (-\offset,\distV){{\Large $\pbTc$}};
  
	\node (umsscf) at (1.5*\offset,2*\distV){{\Large $\umsscfc$}};
 
 	\node (msscf) at (-1.5*\offset,2*\distV){{\Large $\msscfc$}};

\draw[->,dotted, gray, thick, opacity=0.6] (umsscf) -- (msscf);
   \draw[->,dotted, gray, thick, opacity=0.6] (upbT) --  (pbT);

\draw[->,very thick, dashed] (pbv) -- node [above, left] {\bl{\textbf{Lem}~}\ref{thm:pbv_to_pb0_log}} (pbT);

\draw[->,very thick, dashed] (pbv) -- node [above, right] {\bl{\textbf{Lem}~}\ref{thm:pbv_to_upb0_log}}(upbT);

 \draw[->, thick] (pbT) -- node [above, left] {Claim~\ref{thm:upbT_to_umsscf}} (msscf);

        \draw[->,thick] (upbT) -- node [above, right] { \bl{Claim~}\ref{thm:upbT_to_umsscf}}(umsscf);

    \draw[->,dashed, gray] (msscf) to [out=180, in=180]  node [above, left] {Claim~\ref{cl:msscf_to_pb}}  (pbv);


    \draw[->,very thick, dashed] (\ypomnimaX-1, \ypomnimaY+0.8) -- node[above] {{\footnotesize Main Lemma ($\log$ factors)}} (\ypomnimaX+1, \ypomnimaY+0.8);
    
        \draw[->,thick] (\ypomnimaX-1, \ypomnimaY) -- node[above] {{\footnotesize Claim (const. factors)}} (\ypomnimaX+1, \ypomnimaY);
        
    \draw[->,dashed,thick, gray] (\ypomnimaX-1, \ypomnimaY-0.8) -- node[above] {{\footnotesize Minor Claim}} (\ypomnimaX+1, \ypomnimaY-0.8);
    
	  \draw[->,dotted, gray, thick, opacity=0.6] (\ypomnimaX-1, \ypomnimaY-1.6) -- node[above] {{\footnotesize Subproblem}} (\ypomnimaX+1, \ypomnimaY-1.6);

\end{tikzpicture}
    \caption{Reductions shown in this section. The solid lines are part of Corollary~\ref{cor:pb_to_mssc_costs}.}
    \label{fig:pb_and_mssc_costs}
\end{figure}

All the results of this section hold as they are when we change all versions to incorporate costs. We did not use the fact that the costs are unit in any of the proofs of Claim~\ref{cl:msscf_to_pb}, Claim~\ref{thm:upbT_to_umsscf} or Lemmas~\ref{thm:pbv_to_pb0_log}, \ref{thm:pbv_to_upb0_log}. We formally restate the main theorem of Section~\ref{sec:pb_to_mssc} as the following corollary, where the only change is that it now holds for the cost versions of the problems.

\begin{corollary}[\pbText{} to $\msscf$ with non-unit costs] \label{cor:pb_to_mssc_costs}
If there exists an $a(n,m)$ approximation algorithm for $\msscfc$ then there exists a $O(\alpha(n+m, m^2) \log \alpha(n+m, m^2))$-approximation for $\pbvc$. The same result holds if the initial algorithm is for $\umsscfc$.
\end{corollary}

\subsection{Connecting $\msscf$ and Optimal Decision Tree}
In this section the reduction of Theorem~\ref{thm:udt_to_pbT} uses the fact that the costs are uniform. However we can easily circumvent this and obtain corollary~\ref{cor:udt_to_pbT_costs}. Using this, the results for the non-unit costs versions are summarized in Figure~\ref{fig:mssc_and_odt_costs}.

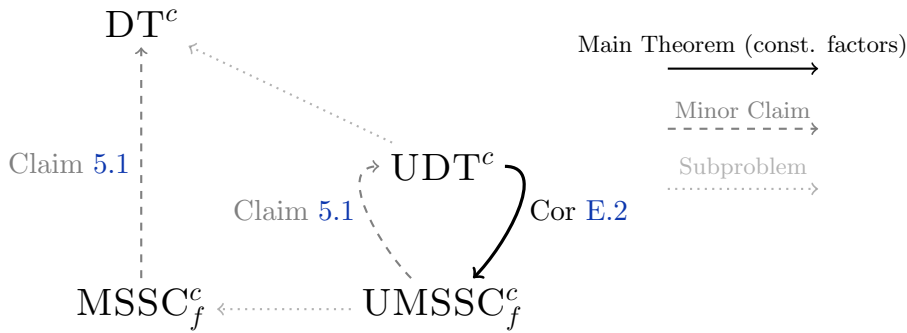
\begin{figure}[H]
    \centering
    \newcommand{\bl}[1]{\textcolor{black}{#1}}
\pgfmathsetmacro{\dist}{3}
\pgfmathsetmacro{\distV}{1.9}
\pgfmathsetmacro{\offset}{1}
\pgfmathsetmacro{\ypomnimaX}{6}
\pgfmathsetmacro{\ypomnimaY}{7}
\begin{tikzpicture}

	\node (umsscfc) at (2*\offset,2*\distV){{\Large $\umsscfc$}};
 	\node (msscfc) at (-2*\offset,2*\distV){{\Large $\msscfc$}};

	\node (udtc) at (2*\offset, 3*\distV){{\Large $\udtc$}};
 	\node (odtc) at (-2*\offset, 4*\distV){{\Large $\odtc$}};
	\draw[->,dotted, gray, thick, opacity=0.6] (umsscfc) -- (msscfc);
 
\draw[->,dotted, gray, thick, opacity=0.6] (udtc) --  (odtc);

 
    \draw[->,dashed, thick, gray] (msscfc) -- node [above, left] {Claim~\ref{thm:msscf_to_odt}} (odtc); 
    
    \draw[->, dashed, thick, gray] (umsscfc) to [out=135,in=180]  node [above, left] {Claim~\ref{thm:msscf_to_odt}} (udtc);

    \draw[->,very thick] (udtc) to [in=45, out=0] node [above, right] {\bl{Cor~}\ref{cor:udt_to_pbT_costs}}  (umsscfc);

    
        \draw[->,thick] (\ypomnimaX-1, \ypomnimaY) -- node[above] {{\footnotesize Main Theorem (const. factors)}} (\ypomnimaX+1, \ypomnimaY);
        
    \draw[->,dashed,thick, gray] (\ypomnimaX-1, \ypomnimaY-0.8) -- node[above] {{\footnotesize Minor Claim}} (\ypomnimaX+1, \ypomnimaY-0.8);
    
	  \draw[->,dotted, gray, thick, opacity=0.6] (\ypomnimaX-1, \ypomnimaY-1.6) -- node[above] {{\footnotesize Subproblem}} (\ypomnimaX+1, \ypomnimaY-1.6);

\end{tikzpicture}
    \caption{Summary of reductions for non unit cost boxes.}
    \label{fig:mssc_and_odt_costs}
\end{figure}

\begin{corollary}[\udtText{} with costs to $\umsscfc$]\label{cor:udt_to_pbT_costs}
Given an $\alpha(m,n)$-approximation algorithm for  $\umsscfc$ then there exists an
	$O(\alpha(n+m, m))$-approximation algorithm for $\udtc$.
\end{corollary}
\begin{proof}
    The proof follows exactly the same way as the proof of Theorem~\ref{thm:udt_to_pbT} with one change: the cost of an isolating element is the minimum cost test needed to isolate $s_i$ from scenario $s_k$ where $s_k$ is the scenario that maximizes this quantity. Formally, if $c(i,k) = \min\{c_j | T_j(i) \not = T_j(k)\}$, then $c(B^i) = \max_{k \in [m]} c(i,k)$. The reduction follows the exact steps as the one we described in Section~\ref{sec:apn_mssc_and_odt}.
\end{proof}


\end{document}